\documentclass[accepted]{uai2022}

\usepackage[american]{babel}

\usepackage{natbib}
 \usepackage{booktabs}
\usepackage[utf8]{inputenc}
\usepackage{hyperref}
\usepackage{amsmath,amssymb,amsfonts,mathrsfs, amsthm}
\usepackage{mathtools}
\usepackage{amsthm}
\usepackage{dsfont}
\usepackage{bm}
\usepackage{colortbl}
\usepackage{fullpage}
\usepackage[ruled, linesnumbered]{algorithm2e}
\usepackage{multirow,booktabs,bigdelim}
\usepackage{caption}
\usepackage{comment}
\usepackage{graphicx}
\usepackage{placeins}
\usepackage{siunitx}
\usepackage{url}
\usepackage{enumitem}
\usepackage{layouts}
\usepackage[rgb]{xcolor}
\usepackage{scalerel}
\usepackage{tikz}
\usepackage{tkz-graph}
\usepackage{authblk}
\usetikzlibrary{shapes.geometric}
\usetikzlibrary{backgrounds}
\usetikzlibrary{arrows.meta}
\usepackage[framemethod=TikZ]{mdframed}
\usepackage{xr}
\theoremstyle{plain}
\newtheorem{theorem}{Theorem} 
\newtheorem{example}[theorem]{Example}

\newtheorem{proposition}[theorem]{Proposition}

\theoremstyle{definition}

\theoremstyle{remark}
\newtheorem{remark}[theorem]{Remark}

\newcommand{\spaceIV}{\texttt{spaceIV} }

\newcommand{\R}{\mathbb{R}}

\DeclareMathOperator*{\argmin}{arg\,min}

\DeclarePairedDelimiterX{\norm}[1]{\lVert}{\rVert}{#1}
\DeclarePairedDelimiterX{\abs}[1]{\lvert}{\rvert}{#1}

\renewcommand{\epsilon}{\varepsilon}

\newcommand{\rank}[1]{\operatorname{Rank}\!\left(#1\right)}
\newcommand{\im}[1]{\operatorname{Im}\!\left(#1\right)}
\newcommand{\nullspace}[1]{\operatorname{Null}\!\left(#1\right)}
\newcommand{\dimension}{\operatorname{dim}}

\newcommand{\independent}{\perp\!\!\!\perp}

\newcommand{\cov}{\operatorname{Cov}}
\newcommand{\vI}{\operatorname{Id}}

\newcommand{\IPA}{\operatorname{PA}_I}
\newcommand{\PA}{\operatorname{PA}}
\newcommand{\ND}{\operatorname{ND}}

\newcommand{\DE}{\operatorname{DE}}
\newcommand{\AN}{\operatorname{AN}}

\newcommand\ti[1]{{\tilde{#1}}}

\newcommand{\supp}{\operatorname{supp}}

\definecolor{col1}{RGB}{88,140,126}
\definecolor{col2}{RGB}{242,227,148}
\definecolor{col3}{RGB}{242,174,114}
\definecolor{col4}{RGB}{217,100,89}
\definecolor{col5}{RGB}{140,70,70}
\colorlet{lightgray}{black!15}

\title{Identifiability of Sparse Causal Effects using Instrumental Variables}

\author[1, *]{\href{mailto:<np@math.ku.dk>}{Niklas~Pfister}{}}
\author[1, *]{\href{mailto:<jonas.peters@math.ku.dk>}{Jonas~Peters}{}}

\affil[1]{%
  Department of Mathematical Sciences\\
  University of Copenhagen\\
  Denmark
}
\affil[*]{%
  Authors contributed equally.
}

\begin{document}

\maketitle

\begin{abstract}
  Exogenous heterogeneity, for example, in the form of instrumental
  variables can help us learn a system's underlying causal structure
  and predict the outcome of unseen intervention experiments.  In this
  paper, we consider linear models in which the causal effect from
  covariates $X$ on a response $Y$ is sparse. We provide conditions
  under which the causal coefficient becomes identifiable from the
  observed distribution. These conditions can be satisfied even if the
  number of instruments is as small as the number of causal parents.
  We also develop graphical criteria under which identifiability holds
  with probability one if the edge coefficients are sampled randomly
  from a distribution that is absolutely continuous with respect to
  Lebesgue measure and $Y$ is childless.  As an estimator, we propose
  \spaceIV and prove that it consistently estimates the causal effect
  if the model is identifiable and evaluate its performance on
  simulated data.  If identifiability does not hold, we show that it
  may still be possible to recover a subset of the causal parents.
\end{abstract}

\section{Introduction}

Instrumental variables \citep{Wright1928, angrist1995identification,
  Newey2013} allow us to consistently estimate causal effects from
covariates $X$ on a response $Y$ even if the covariates and response
are connected through hidden confounding.  These approaches usually
rely on identifying moment equations such as
$\cov[I, Y - X^\top \beta] = 0$ with $I$ being the instrumental
variable (IV).  Under some assumptions such as the exclusion
restriction, this equation is satisfied for the true causal
coefficient $\beta = \beta^*$; in a linear setting, for example, this
is the case if we can write $Y = X^\top \beta^* + g(H, \varepsilon^Y)$
with $H, \varepsilon^Y$ being independent of $X$ and $I$
and $\varepsilon^Y$ independent of $X$.  Identifiability of
$\beta^*$, however, requires that the moment equation is not satisfied
for any other $\beta \neq \beta^*$. Formally, this condition is often
written as a rank condition on the covariance between $I$ and $X$,
which implies that the dimension of $I$ must be at least as large as
the number of components of $X$.

In this work, we consider the case where the causal coefficient
$\beta^*$ is assumed to be sparse.  This assumption allows us to relax
existing identifiability conditions: it is, for example, possible to
identify $\beta^*$ even if there are much less instruments than
covariates.  Our results are proved in the context of linear
structural causal models (SCMs) \citep{Pearl2009, Bongers2021}, that
is, we also assume linearity among the $X$ variables.  We prove
sufficient conditions for identifiability of $\beta^*$ that are based
on rank conditions of the matrix of causal effects from $I$ on the
parents of $Y$.  We then investigate for which graphical structures we
can expect such conditions to hold.  Consider, for example, the graph
shown in Figure~\ref{fig:ex3-marg}.  Square nodes represent
instruments, and hidden variables between variables in $X \cup \{Y\}$
can exist but are not drawn (we formally introduce such graphs in
Section~\ref{sec:graphs}). Sparse identifiability in this graph is not
obvious: Is the causal effect from the parents of $Y$ to $Y$
generically identifiable if the true underlying and unknown graph is
the one shown (including the two dashed edges)? And what about the
graph excluding the two dashed edges?
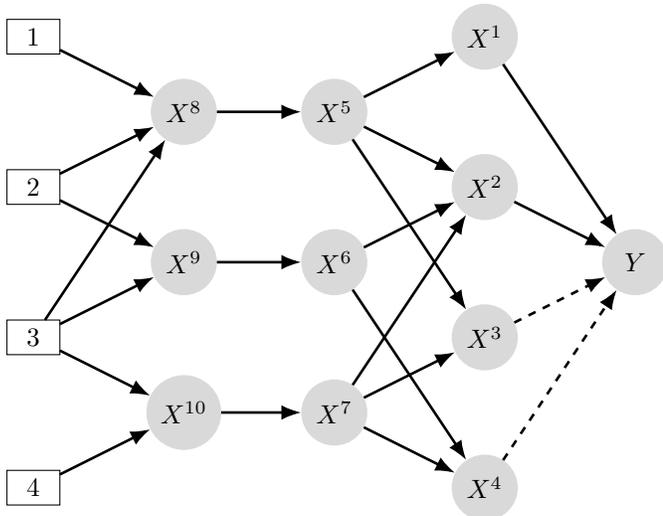
\begin{figure}[ht]
  \centerline{
  \begin{tikzpicture}[scale=1]
    \tikzstyle{VertexStyle} = [shape = circle, minimum width = 2.5em, fill=lightgray]
    \Vertex[Math,L=Y,x=0,y=0]{Y}
    \Vertex[Math,L=X^1,x=-2,y=3]{X1}
    \Vertex[Math,L=X^2,x=-2,y=1]{X2}
    \Vertex[Math,L=X^3,x=-2,y=-1]{X3}
    \Vertex[Math,L=X^4,x=-2,y=-3]{X4}
    \Vertex[Math,L=X^5,x=-4,y=2]{X5}
    \Vertex[Math,L=X^6,x=-4,y=0]{X6}
    \Vertex[Math,L=X^7,x=-4,y=-2]{X7}
    \Vertex[Math,L=X^8,x=-6,y=2]{X8}
    \Vertex[Math,L=X^9,x=-6,y=0]{X9}
    \Vertex[Math,L=X^{10},x=-6,y=-2]{X10}
    \tikzstyle{VertexStyle} = [draw, shape = rectangle, minimum
    width=2em]
    \Vertex[Math,L=1,x=-8.0,y=3]{1}
    \Vertex[Math,L=2,x=-8.0,y=1]{2}
    \Vertex[Math,L=3,x=-8.0,y=-1]{3}
    \Vertex[Math,L=4,x=-8.0,y=-3]{4}
    \tikzset{EdgeStyle/.append style = {-Latex, line width=1}}
    \Edge(1)(X8)
    \Edge(2)(X8)
    \Edge(2)(X9)
    \Edge(3)(X8)
    \Edge(3)(X9)
    \Edge(3)(X10)
    \Edge(4)(X10)
    \Edge(X1)(Y)
    \Edge(X2)(Y)
    \Edge(X8)(X5)
    \Edge(X9)(X6)
    \Edge(X10)(X7)
    \Edge(X5)(X1)
    \Edge(X5)(X2)
    \Edge(X5)(X3)
    \Edge(X6)(X2)
    \Edge(X6)(X4)
    \Edge(X7)(X2)
    \Edge(X7)(X3)
    \Edge(X7)(X4)
    \tikzset{EdgeStyle/.append style = {-Latex, line width=1, dashed}}
    \Edge(X3)(Y)
    \Edge(X4)(Y)
\end{tikzpicture}}
\caption{Graphical representation of two linear SCMs, as described in
  Section~\ref{sec:graphs} (hidden variables between $X$ and $Y$
  variables exist but are not drawn).  If the data come from a system
  corresponding to the unknown graph including (or excluding) dashed
  edges, can we identify the causal effect from $X$ to $Y$ from the
  joint distribution over $I$, $X$, and $Y$? These questions are
  discussed in Example~\ref{ex:graph}.  }\label{fig:ex3-marg}
\end{figure}
We translate the rank conditions for identifiability to structural
SCMs whose coefficients are drawn randomly from a distribution that is
absolutely continuous with respect to Lebesgue measure. This allows us
to develop graphical criteria that can answer these questions.

If identifiability holds, the causal effect can be estimated from
data. We propose an estimator called \spaceIV (`\textbf{spa}rse
\textbf{c}ausal \textbf{e}ffect \textbf{IV}').  It is based on the
limited information maximum likelihood (LIML) estimator
\citep{anderson1949estimation, amemiya1985advanced}.  This estimator
has similar properties as the two stage least squares estimator and
has the same asymptotic normal distribution, for example
\citep{mariano2001simultaneous}.  But as it minimizes the
Anderson-Rubin test statistic, it allows us to prove theoretical
guarantees.  We evaluate the performance of \spaceIV on simulated
data. If identifiability does not hold, we prove that it may still be
possible to identify a subset of the causal parents of $Y$.

Numerous extensions to the classical linear instrumental variable
setting have been proposed.  For example, nonlinear effects
\citep{Imbens2009, Dunker2014, Torgovitsky2015, Loh2019,
  Christiansen2020DG} have been considered, often in relation with
higher order moment equations \citep{Hartford2017, Singh2019,
  Bennett2019, Muandet2020, Saengkyongam2022}.  Furthermore,
\citet{Belloni2012, Mckeigue2010} assume that the effect from the
instruments on the covariates is sparse. For example, it has been
shown that consistent estimators exist if at least half of all
instruments are valid \citep{Kang2016}.  To the best of our knowledge,
while existing work considers sparsity constraints between the
instruments and the covariates (`first stage'), the assumption of a
sparse causal effect (`second stage') and its benefits has not yet
been analyzed.

Our paper is structured as follows. Section~\ref{sec:model} introduces
the formal setup. Section~\ref{sec:identifsparse} presents the main
identifiability result for sparse causal effect models and
Section~\ref{sec:graphchar} develops the corresponding graphical
criteria.  Section~\ref{sec:method} introduces the estimator \spaceIV
and Section~\ref{sec:experiments} includes simulation
experiments. Code is attached as supplementary material.

\section{IV Models with Shift Interventions}\label{sec:model}

Consider the following structural causal model (SCM)
\begin{equation}
\begin{split}
    X &:= BX + AI + h(H, \epsilon^X)\\ 
    Y &:= X^{\top}\beta^* + g(H, \epsilon^Y),
\end{split}
\label{eq:scm}
\end{equation}
where $h$ and $g$ are arbitrary measurable functions and
$\vI-B$\footnote{Here $\vI$ denotes the identity matrix.} is
invertible. Here, $X \in \R^d$ denotes the observed variables,
$H \in \R^q$ the unobserved variables, $I \in \R^m$ the instrumental
variables (following an $m$-dimensional distribution, which is not
modelled explicitly), $Y \in \R$ the response and $I$, $H$,
$\epsilon^X$ and $\epsilon^Y$ are jointly independent and assume that
the covariates are non-descendants of $Y$ (see also
Remark~\ref{rem:children}).  In contrast to classical IV settings, we
thus explicitly model the causal effects of the instruments $I$ on the
predictor variables $X$.  Throughout the paper, we assume that
$\cov[I]$ is invertible.  We assume that we have access to an i.i.d.\
data set $(X_1, Y_1, I_1), \ldots, (X_n, Y_n, I_n)$ sampled from the
induced distribution and are interested in estimating the causal
effect $\beta^*$.  We call the set of non-zero components of $\beta^*$
the parents of $Y$ and denote it by $\PA(Y)$.

Our model covers the case, where we observe data from $m$ different
experiments, each of which corresponds to a fixed intervention
shift. More precisely, we can choose $I$ such that for all
$k \in \{1, \ldots, m\}$, we have $P(I = e_k) = 1/m$, with $e_k$,
$k \in \{1, \ldots, m\}$, being the $k$-th unit vector in
$\mathbb{R}^m$. Here, each column in the matrix $A$ specifies a
different experiment in which (a subset of) the $X$ variables is
shifted by the amount specified in that column.

\subsection{Graphical Representation} \label{sec:graphs}

Given a data generating process of the form~\eqref{eq:scm}, we
represent it graphically as follows: Each of the $d$
components\footnote{In a slight abuse of notation, we sometimes
  identify each component with its index.} of $X$ is represented by a
node, which we call a \emph{prediction node}. There is a directed edge
from $X^i$ to $X^j$ if and only if $B_{ji} \neq 0$.  In addition, we
represent the $k$th component of $I$ by a square node with label
`$k$', which we call \emph{instrument node}. There is a directed edge
from $k$ to $X^j$ if and only if $A_{j,k} \neq 0$.  (There are no
connections between instrument nodes, even though they may be
dependent.)  Finally, we represent the response $Y$ with the same node
style as is used for the predictors and include a directed edge from
$X^j$ to $Y$ if and only if $\beta^j\neq 0$.  In the graph, we do not
represent hidden variables (even though they are allowed to exist).
Consequently, such graphs do not satisfy the Markov condition
\citep[e.g.,][]{Lauritzen1996}.

\begin{example} \label{ex:1}
Consider an SCM of the following form
\begin{align}
\left(
\begin{matrix}
X^1\\
X^2\\
X^3
\end{matrix} 
\right)
&:= 
\left(
\begin{matrix}
0\\
b_{21} X^1\\
0
\end{matrix} 
\right)
+ 
\left(
\begin{matrix}
a_{11} & a_{12} \\
0 & a_{22} \\
0 & a_{32} \\
\end{matrix} 
\right)
\left(
\begin{matrix}
I^1\\
I^2
\end{matrix} 
\right)  + h(H, \epsilon^X) 
\nonumber\\
Y &:= 
\left(
\begin{matrix}
X^1 & X^2 & X^3
\end{matrix} 
\right)
\left(
\begin{matrix}
0\\
\beta^*_{2}\\
0
\end{matrix} 
\right)
 + g(H, \epsilon^Y), \label{eq:ex1}
\end{align}
where 
$I^1$, $I^2$, $H$, $\epsilon^Y$, $\epsilon^X$ are jointly independent. 
Figure~\ref{fig:ex1} shows the corresponding graphical representation.
\begin{figure}[t]
  \centering
  \begin{tikzpicture}[scale=1]
    \tikzstyle{VertexStyle} = [shape = circle, minimum width =
    2.5em, fill=lightgray]
    \Vertex[Math,L=Y,x=0,y=0]{Y}
    \Vertex[Math,L=X^1,x=-6,y=0]{X1}
    \Vertex[Math,L=X^2,x=-3,y=0]{X2}
    \Vertex[Math,L=X^3,x=-1.5,y=2]{X3}
    \tikzstyle{VertexStyle} = [draw, shape = rectangle, minimum
    width=2em]
    \Vertex[Math,L=1,x=-6.0,y=2]{1}
    \Vertex[Math,L=2,x=-4.5,y=2]{2}
    \tikzstyle{VertexStyle} = [draw, dashed, shape = circle, minimum
    width=2.5em]
    \tikzset{EdgeStyle/.append style = {-Latex, line width=1}}
    \Edge(1)(X1)
    \Edge(X1)(X2)
    \Edge(X2)(Y)
    \Edge(2)(X1)
    \Edge(2)(X3)
    \Edge(2)(X2)
  \end{tikzpicture}
  \caption{Graphical representation of Example~\ref{ex:1}, as
    described in Section~\ref{sec:graphs} (the hidden variable $H$ is
    omitted).  }\label{fig:ex1}
\end{figure}
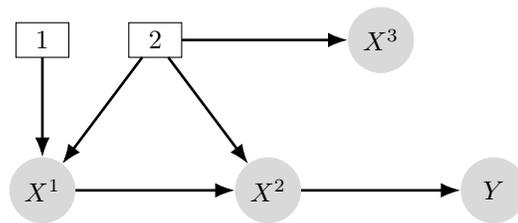
\end{example}

\section{Identifiability in Sparse-Effect IV Models}

Consider a data generating process of the form~\eqref{eq:scm}.
Because the intervention $I$ does not directly enter the structural
assignment of $Y$ and $(H, \epsilon^Y, I)$ are jointly independent,
the causal coefficient $\beta^*$ satisfies the \emph{moment condition}
\begin{equation}
  \label{eq:moment_eq}
  \cov\left(I, Y-X^{\top}\beta^*\right)=0.
\end{equation}
The solution space of the moment condition is given by
\begin{equation*}
  \mathcal{B}:=\{\beta\in\R^d\,\vert\, \cov(I,
  X)\beta=\cov(I, Y)\}.
\end{equation*}
It can be shown that this is a $(d-\rank{A})$-dimensional space. The
true causal coefficient $\beta^*$ is therefore identified by
\eqref{eq:moment_eq} if and only if $\rank{A}=d$. This directly
implies that the number of instruments needs to be greater or equal to
the number of predictors, a well-known necessary condition for
identifiability in the linear IV model.

In this work, we investigate the case where $\mathcal{B}$ is allowed
to be non-degenerate.  To analyse conditions for identifiability, we
define the $(m \times d)$-matrix
\begin{equation}
  \label{eq:Cmatrix}
  C\coloneqq A^{\top}(\vI-B)^{-\top}.
\end{equation}
The entry $C_{i,j}$ corresponds to the the $i$-th component of the
total causal effect from $I$ onto $X^j$ in the SCM given in
\eqref{eq:scm}. This entry correspond to summing over all directed
paths from instrument node $i$ to $X^j$ and for each path, multiplying
the coefficients. The matrix $C$ will play a central role when
analyzing identifiability.  For example, using the matrix $C$,
Proposition~\ref{thm:partial_identifiability} characterizes settings
under which individual components of the causal coefficient $\beta^*$
are identifiable. This result does not require any additional
assumptions on the underlying model. In
Section~\ref{sec:identifsparse}, we then show that if the causal
coefficient $\beta^*$ is sparse (i.e., it contains many zeros) it can
still be identifiable even if $\mathcal{B}$ is non-degenerate.
\begin{proposition}[Partial identifiability of causal coefficient]
  \label{thm:partial_identifiability}
  Consider a data generating process of the form~\eqref{eq:scm}. Then, for all $j\in\{1,\ldots,d\}$ it holds that
  \begin{equation*}
      \beta^*_j\text{ is identifiable by \eqref{eq:moment_eq}}
      \quad\Leftrightarrow\quad
      \nullspace{C}_j=\{0\},
  \end{equation*}
  where $\nullspace{C}_j$ denotes the $j$-th coordinate of the null
  space of $C$. Moreover, whenever $\operatorname{Null}(C)_j=\{0\}$ it
  holds that $\beta^*_j=(\cov(I,X)^{\dagger}\cov(I,Y))_j$, where
  $(\cdot)^{\dagger}$ denotes the Moore-Penrose inverse.
\end{proposition}
The proof can be found in Appendix~\ref{app:partial}.

\subsection{Identifiability of sparse causal coefficients} \label{sec:identifsparse}

We have argued that the causal parameter is in general not fully
identified by the moment condition \eqref{eq:moment_eq}.  However, we
can obtain identifiability by additionally assuming that the causal
coefficient $\beta^*$ is sparse. To make this more precise, consider
the following optimization
\begin{equation}
  \label{eq:sparse_optim}
  \min_{\beta\in\mathcal{B}}\,\norm{\beta}_0.
\end{equation}
As we will see below, under mild conditions on the interventions $I$,
the causal coefficient $\beta^*$ is a unique solution to this problem.

We now make the following assumptions\footnote{Here we use the
  convention that for a matrix $D\in\R^{m\times d}$ and a subset
  $S\subseteq\{1,\ldots,d\}$ the subindexed matrix $D_S$ corresponds
  to the $m \times |S|$-submatrix of $D$ consisting of all columns
  that are indexed by $S$ and $\im{D}$ denotes the image of $D$.}.
\begin{enumerate}
\item[(A1)] It holds that $\rank{C_{\PA(Y)}}=\abs{\PA(Y)}$.
\item[(A2)] For all $S\subseteq\{1,\ldots,d\}$ it holds that
  \begin{align*}
    \left.
    \begin{aligned}
      &\rank{C_S}\leq\rank{C_{\PA(Y)}} \text{ and}\\
      &\im{C_S}\neq\im{C_{\PA(Y)}}
    \end{aligned}\right\}
    \text{ implies }\\
    \left\{
    \forall w\in\R^{\abs{S}}:\quad C_Sw\neq C_{\PA(Y)}\beta^{*}_{\PA(Y)}\right..
  \end{align*}
\item[(A3)] For all $S\subseteq\{1,\ldots,d\}$ with
  $\abs{S}=\abs{\PA(Y)}$ and $S\neq\PA(Y)$ 
  we have $\im{C_{S}}\neq\im{C_{\PA(Y)}}$.
\end{enumerate}
(A1) is a necessary assumption in order to identify $\beta^*$; it
guarantees that an IV regression based on the true parent set $\PA(Y)$
identifies the correct coefficients. (A2) is an assumption on the
underlying causal model that ensures that certain types of
cancellation cannot occur. It is a rather mild assumption in the
following sense: if one considers the true causal parameter $\beta^*$
as randomly drawn from a distribution absolutely continuous with
respect to Lebesgue measure it would almost surely lead to a system
that satisfies (A2) (see Proposition~\ref{prop:random_proj} in
Appendix~\ref{app:additional}). As shown in the following theorem,
(A1) and (A2) are sufficient to ensure that $\beta^*$
solves~\eqref{eq:sparse_optim}. Additionally assuming (A3) ensures
that the solution is unique; it can be seen as requiring an extra
level of heterogeneity in how the interventions affect the system (see
also Section~\ref{sec:graphchar}).

\begin{theorem}[Identifiability of sparse causal parameters]
  \label{thm:sparse_identifiability}
  Consider a data generating process of the form~\eqref{eq:scm}. If
  (A1) and (A2) hold, then $\beta^*$ is a solution to
  \eqref{eq:sparse_optim}. Moreover, if, in addition, (A3) holds, then
  $\beta^*$ is the unique solution.
\end{theorem}

\begin{proof}
  We use the notation $\xi^X\coloneqq h(H, \epsilon^X)$ and
  $\xi^Y\coloneqq g(H, \epsilon^X)$. Then \eqref{eq:scm} and the
  assumption of joint independence of $I$, $\xi^X$, and $\xi^Y$ imply
  that
  \begin{align}
    \cov[I, X]
    &=\cov\left[I, (\vI-B)^{-1}(AI+\xi^X)\right]\nonumber\\
    &=\cov[I]A^{\top}(\vI-B)^{-\top}\label{eq:covXXpart}.
  \end{align}
  Similarly, we get that
  \begin{align}
    \cov[I, Y]
    &=\cov\left[I, (AI+\xi^X)^{\top}(\vI-B)^{-\top}\beta^*+\xi_Y\right]\nonumber\\
    &=\cov\left[I, {\beta^*}^{\top}(\vI-B)^{-1}(AI+\xi^X)+\xi_Y\right]\nonumber\\
    &=\cov[I]A^{\top}(\vI-B)^{-\top}\beta^*.\label{eq:covXYpart}
  \end{align}
  Hence, for any $\ti{\beta}\in\mathcal{B}$, using the definition of
  $\mathcal{B}$ and combining \eqref{eq:covXXpart} and
  \eqref{eq:covXYpart} we get that
  \begin{equation*}
    \cov[I]C\ti{\beta}=\cov[I]C\beta^*.
  \end{equation*}
  Here, we used the definition of $C$ in \eqref{eq:Cmatrix}. 
  As $\cov[I]$ is invertible, we get
  \begin{equation}
    \label{eq:reduced_moment_cond_non_sparse}
    C\ti{\beta}=C\beta^*,
  \end{equation}
  Furthermore, it holds for\footnote{The support of a vector
      is defined as the set of indices of non-zero elements.}
    $S=\operatorname{supp}(\ti{\beta})$ that
  \begin{equation}
    \label{eq:reduced_moment_cond}
    C_S\ti{\beta}_{S}=C_{\PA(Y)}\beta^*_{\PA(Y)}.
  \end{equation}

  We now prove the first part of the theorem. Assume (A1) and (A2) are
  satisfied. We want to show that
  \begin{equation}
    \label{eq:optimization_argmin}
    \beta^*\in\argmin_{\beta\in\mathcal{B}}\norm{\beta}_0.
  \end{equation}
  Since $\beta^*\in\mathcal{B}$, it is sufficient to show that for all
  $\ti{\beta}\in\mathcal{B}$ it holds that
  $\norm{\ti{\beta}}_0\geq \abs{\PA(Y)}$. To this
  end, fix $\ti{\beta}\in\mathcal{B}$ and set
  $S=\operatorname{supp}(\ti{\beta})$. For the sake of contradiction
  assume $\abs{S}<\abs{\PA(Y)}$, then using (A1) we get that
  \begin{align*}
    \rank{C_{\PA(Y)}}&=
                       \dimension(\im{C_{\PA(Y)}})=\abs{\PA(Y)}
    \\ &> S \geq \dimension(\im{C_S})
         =\rank{C_{S}}.
  \end{align*}
  This implies $ \rank{C_{\PA(Y)}} \geq \rank{C_{S}}$ and
  $\im{C_{\PA(Y)}} \neq \im{C_{S}}$. Thus, by (A2), this contradicts
  \eqref{eq:reduced_moment_cond}. This completes the first part of
  the proof.

  Next, we prove the second part of the theorem. Assume that (A1),
  (A2) and (A3) are satisfied. By the previous part of the proof, we
  have seen that $\beta^*$ satisfies \eqref{eq:optimization_argmin}.
  It therefore only remains to show that there is no other
  solution. Assume for the sake of contradiction that there exists
  $\ti{\beta}\in\mathcal{B}$ with $S:=\operatorname{supp}(\ti{\beta})$
  such that $\abs{S}=\abs{\PA(Y)}$ and $S\neq\PA(Y)$. Then by (A3) we
  have $\im{C_{S}}\neq\im{C_{\PA(Y)}}$.  By (A1) it holds that
  \begin{equation*}
    \rank{C_{\PA(Y)}}=\abs{\PA(Y)}=\abs{S}\geq\rank{C_S}.
  \end{equation*}
  Hence, together with the condition $\im{C_{S}}\neq\im{C_{\PA(Y)}}$
  we can use (A2) to get a contradiction to
  \eqref{eq:reduced_moment_cond}.  This completes the proof of
  Theorem~\ref{thm:sparse_identifiability}.
\end{proof}

Section~\ref{sec:subset_identifiability} shows how one can identify a
subset of the causal parents under even milder
conditions. Remark~\ref{rem:children} discusses the case where the
covariates can also be descendants of $Y$.

\section{Graphical Characterization} \label{sec:graphchar}

We now formulate the identifiability result from
Section~\ref{sec:identifsparse} in graphical terms.  Suppose we are
given a data generating process of the form~\eqref{eq:scm} with
corresponding graph $\mathcal{G}$ (as described in
Section~\ref{sec:graphs}), which in this section is assumed to be
acyclic.  The parents of $Y$ are denoted by $\PA(Y)$ and correspond to
the non-zero entries of $\beta^*$. Moreover, for any set
$S\subseteq\{1,\ldots,d\}$, we define the set of all
\emph{intervention ancestors} of variables in $S$ as
$$
\AN_I[S]\coloneqq\{j\in\{1,\ldots,m\}\,\vert\, j \in
\AN(S)\}.\footnote{$\AN(S)$ denotes the ancestor set of
    $S$ consisting of all nodes with a directed path to a node in
    $S$. Throughout the paper, we use the convention that a node is
    contained in the set of its ancestors.}
$$ 
This set contains the instrument nodes that are ancestors of $S$.

We can now state the following graphical assumptions.
\begin{enumerate}
\item[(B1)] There are at least $|\PA(Y)|$ disjoint directed paths (not
  sharing any node) from $I$ to $\PA(Y)$.

\item[(B2)] The non-zero coefficients of the causal coefficient
  $\beta^*_{\PA(Y)}\in\R^{\abs{\PA(Y)}}$ and the non-zero entries of
  $A$ and $B$ are randomly drawn from a distribution $\mu$ which is
  absolutely continuous with respect to Lebesgue measure (and are
  independent of the other variables).
\item[(B3)] For all $S\subseteq\{1,\ldots,d\}$ with
  $\abs{S}=\abs{\PA(Y)}$ and $S\neq\PA(Y)$ at least one of the
  following conditions is satisfied
  \begin{itemize}
  \item[(i)] 
    $\AN_I[S]\neq \AN_I[\PA(Y)]$.
  \item[(ii)] The smallest set $T$ of nodes such that all directed
    paths from $I$ to $\PA(Y)$ and from $I$ to $S$ go through $T$ is
    of size at least $|\PA(Y)|+1$.
  \end{itemize}
\end{enumerate}
We will see in Theorem~\ref{thm:sparse_identifiability_graph} below
that the causal effect becomes identifiable if (B1)--(B3) hold.  Let
us discuss these assumptions using a few examples.
\begin{example}\label{ex:graph}
\begin{itemize}
\item[(i)] The example from Example~\ref{ex:1} and
  Figure~\ref{fig:ex1} is discussed in Figure~\ref{fig:ex1-marg} in
  Appendix~\ref{app:graphs}.
\item[(ii)] Figure~\ref{fig:ex2-marg} contains another identifiable
  example.
\item[(iii)] We now come back to the example graphs shown in
  Figure~\ref{fig:ex3-marg}.  Consider an SCM with the graph structure
  shown including the dashed edges.  (B1) is violated, as the effect
  of the four instruments is `channelled' through three
  variables. Indeed, here, the causal effect from
  $(X^1, X^2, X^3, X^4)$ on $Y$ is in general non-identifiable -- even
  though all instruments are connected to all causal parents of $Y$
  (the rank of $C_{\PA(Y)}$ is three and therefore too small to
  identify $\beta^*$).
\item[(iv)] Consider an SCM with the graph structure shown in
  Figure~\ref{fig:ex3-marg} (dashed edges not included).  Here, (B1)
  holds.  (B3) is satisfied, too: e.g., for the set
  $S := \{X^3, X^4\}$, we have
  $\AN_I[\{X^3, X^4\}] = \AN_I[\{X^1, X^2\}]$, so (B3) (i) is
  violated, but (B3) (ii) holds (which implies
  $\im{C_{S}}\neq\im{C_{\PA(Y)}}$, see proof of
  Theorem~\ref{thm:sparse_identifiability_graph}): there is no set of
  size two such that all directed paths go through this set (note that
  $|\AN(\{X^3, X^4\})|=3$, for example).  Thus, if additionally (B2)
  holds, then $\beta^*$ is identifiable.
\end{itemize}
\end{example}

A graphical marginalization of graphs similar to the latent projection
\citep{richardson2003markov, Ver93} may help to gain further intuition
about the assumptions.  Consider a subset
$V \subseteq \{1, \ldots, d\}$ of the covariates.  The marginalized
graph $\mathcal{G}^V$ is then constructed from $\mathcal{G}$ by the
following procedure: (i) $\mathcal{G}^V$ consists of all instrument
nodes $k$ from $\mathcal{G}$, all predictor nodes $X^j$ from
$\mathcal{G}$ for which $j\in V$, and node $Y$; (ii) $\mathcal{G}^V$
contains a directed edge from $X^i$ to $X^j$ if and only if
$\mathcal{G}$ contains a directed path from $X^i$ to $X^j$ that does
not have any intermediate nodes in $V$ (e.g., because there are no
intermediate nodes); (iii) $\mathcal{G}^V$ contains a directed edge
from $k$ to $X^j$ if and only if $\mathcal{G}$ contains a directed
path from $k$ to $X^j$ that does not have any intermediate nodes in
$V$ (e.g., because there are no intermediate nodes).  The set
$\IPA^V[U]$ denotes the $\mathcal{G}^V$-parents of $U$ that are
intervention nodes: $\IPA^V[U] := \PA_{\mathcal{G}^V} \cap I$.
Figure~\ref{fig:ex1-marg}, in Appendix~\ref{app:graphs}, shows the
marginalized graph corresponding to Example~\ref{ex:1} and
Figure~\ref{fig:ex1}.

(B1) ensures that there is sufficient heterogeneity coming from
instruments. In particular, there need to be as many instruments as
parents of $Y$ and for all $S \subseteq \PA(Y)$, we have
$|\PA_I^{\PA(Y)}[S]| \geq |S|$. In particular, this implies that for
all $k \in \PA(Y)$, we have $\PA_I^{\PA(Y)}[k] \neq \emptyset$.  In
general, however, this is not sufficient for identifiability (see
Section~\ref{sec:experiments}).
\begin{figure}[t]
  \centerline{
  \begin{tikzpicture}[scale=1]
    \tikzstyle{VertexStyle} = [shape = circle, minimum width = 2.5em, fill=lightgray]
    \Vertex[Math,L=Y,x=0,y=0]{Y}
    \Vertex[Math,L=X^1,x=-2,y=1]{X1}
    \Vertex[Math,L=X^2,x=-2,y=-1]{X2}
    \Vertex[Math,L=X^3,x=-4,y=2]{X3}
    \Vertex[Math,L=X^4,x=-4,y=0]{X4}
    \Vertex[Math,L=X^5,x=-4,y=-2]{X5}
    \tikzstyle{VertexStyle} = [draw, shape = rectangle, minimum
    width=2em]
    \Vertex[Math,L=1,x=-6.0,y=2]{1}
    \Vertex[Math,L=2,x=-6.0,y=0]{2}
    \Vertex[Math,L=3,x=-6.0,y=-2]{3}
    \tikzset{EdgeStyle/.append style = {-Latex, line width=1}}
    \Edge(1)(X3)
    \Edge(2)(X4)
    \Edge(3)(X5)
    \Edge(X1)(Y)
    \Edge(X2)(Y)
    \Edge(X3)(X1)
    \Edge(X4)(X1)
    \Edge(X4)(X2)
    \Edge(X5)(X2)
  \end{tikzpicture}}
\caption{Graphical representation of an example SCM, as described in
  Section~\ref{sec:graphs} (there may be hidden variables between all
  predictor variables). Here, (B1) holds as there are two distinct
  paths from $I$ to $\PA(Y)$.  (B3) is satisfied, too: $\{X^1, X^5\}$
  and $\{X^3, X^2\}$ are the only sets $S$ violating (B3) (i) (because
  $\AN_I[\{X^1, X^5\}] = \AN_I[\{X^1, X^2\}] = \AN_I[\{X^3, X^2\}]$),
  but they satisfy (B3) (ii).  Thus, given (B2), the causal effect
  from $(X^1, X^2)$ on $Y$ is identifiable -- even though there are
  less instruments than covariates.  }\label{fig:ex2-marg}
\end{figure}
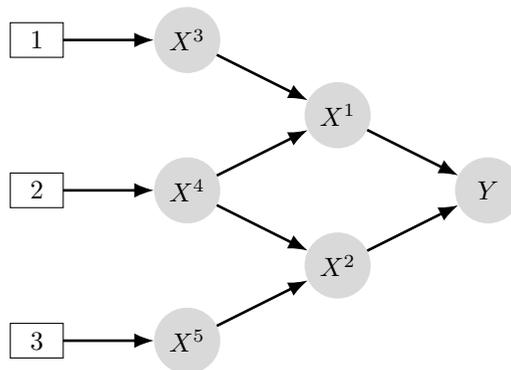

We can now state the graphical version of Theorem~\ref{thm:sparse_identifiability}.  
\begin{theorem}[Identifiability of sparse causal coefficients (graph version)]
  \label{thm:sparse_identifiability_graph}
  Consider a data generating process of the form~\eqref{eq:scm}.  If
  (B1) and (B2) hold, then (A1) and (A2) hold $\mu$-almost surely.
  Moreover, if, in addition, (B3) holds, then (A3) holds $\mu$-almost
  surely.
\end{theorem}
Together with Theorem~\ref{thm:sparse_identifiability} this implies
that under (B1) and (B2), $\beta^*$ is $\mu$-almost surely a solution
to \eqref{eq:sparse_optim} and (B1), (B2), and (B3), it is
$\mu$-almost surely the unique solution.

\begin{proof}
  Regarding (A1): With respect to (A1), consider first the
  marginalization of model~\eqref{eq:scm} over $\PA(Y)$.  To do so, we
  repeatedly substitute $X^j$, $j \in \{1, \ldots, d\}$ with its
  assignment, that is, the corresponding right-hand side
  of~\eqref{eq:scm} and obtain
  \begin{align} \label{eq:margpar}
    \begin{array}{ll}
      X^{\PA(Y)} := 
      {C^\times}^\top
      I + h^\times(H,\epsilon^X).
    \end{array}
  \end{align}
  We then have
  \begin{equation} \label{eq:ctimes}
    C_{\cdot, \PA(Y)} = C^\times,
  \end{equation}
  where $C_{\cdot, \PA(Y)}$ is the matrix constructed from the columns
  of $C$ corresponding to $\PA(Y)$.  Equality~\eqref{eq:ctimes} holds
  by construction: The element of $C_{\cdot, \PA(Y)}$ in row $i$ and
  the column corresponding to $X^j \in \PA(Y)$ equals the $i$-th
  component of the total causal effect from $I$ on $X^j$; this is
  exactly the same in the marginalized model~\eqref{eq:margpar}.  We
  now argue that ${C^\times}$ has full rank $\mu$-almost surely.  To
  do so, we perform a more careful replacement scheme that allows us
  to write
  \begin{equation}
    \label{eq:replsc}    
    X^{\PA(Y)} = C_1 \cdot C_2 \cdot \ldots \cdot C_f \cdot I + h^\times(H,\varepsilon^X).
  \end{equation}
  It then holds that
  $C^\times = (C_1 \cdot C_2 \cdot \ldots \cdot C_f)^\top$.  As a
  first step of the replacement scheme, consider all $X$ nodes on
  directed paths from $I$ to $\PA(Y)$, that is
  $W := \AN(\PA(Y)) \cap \DE(I)$. Among these nodes we consider a
  causal ordering on the induced graph, that is, we choose
  $i_1, \ldots, i_f$ such that for all $k, \ell \in \{1, \ldots, f\}$
  with $k < \ell$, we have
  $X^{i_{\ell}} \in \ND_{\mathcal{G}_{W}}(X^{i_k})$ ($\ND$ denotes the
  "non-descendants"), where $\mathcal{G}_{W}$ is the subgraph of
  $\mathcal{G}$ over nodes in $W$.  We now start from the equation
  $X^{\PA(Y)}=X^{\PA(Y)}$ and replace, on the right-hand side,
  $X^{i_{1}}$ by its structural equation, yielding
  $$
  X^{\PA(Y)} = C_1 \cdot X^{\PA_1} + h_1(H, W^c, \varepsilon^{X_{i_1}}),
  $$ 
  where
  $\PA_1 = \PA(Y) \setminus \{ X^{i_1} \} \cup
  \PA_{\mathcal{G}_W}(X^{i_1})$ and the $h_1$ term collects error
  terms and variables not in $W$.  $C_1$ is a matrix with dimension
  $|\PA(Y)| \times |\PA_1|$.  We did not replace the variables in
  $\PA(Y) \setminus \{ X^{i_1} \}$, the corresponding submatrix in
  $C_1$ is the identity.  All directed paths from $I$ to
  $\PA(Y)$ go through $\PA_1$. Condition~(B1) therefore implies
  $|\PA_1| \geq |\PA_1(Y)|$.  The row corresponding to $X^{i_1}$
  contains the path coefficients from $\PA(X^{i_1})$ to $X^{i_1}$,
  which are $\mu$-almost surely non-zero. Thus, $C_1$ has $\mu$-almost
  surely rank $|\PA(Y)|$.  We now repeatedly (for
  $k \in \{2, \ldots, \ell\}$) substitute the variable $X^{i_k}$ in
  $X^{\PA_{k-1}}$ with its structural equation yielding
  $$
  X^{\PA(Y)} = C_1 \cdot C_2 \cdot \ldots \cdot C_{k} \cdot
  X^{\PA_{k}} + h_k(H, W^c, \varepsilon^{X_{i_1}}),
  $$ 
  where
  $\PA_k = \PA_{k-1} \setminus \{ X^{i_k} \} \cup
  \PA_{\mathcal{G}_W}(X^{i_k})$ and $C_{k}$ contains an identity
  matrix for the submatrix, corresponding to $\PA_k \cap \PA_{k-1}$
  and in the row corresponding to $X^{i_k}$ a vector of coefficients.
  With the same arguments as above, we have that
  $|\PA_k| \geq |\PA(Y)|$. (Indeed, otherwise, all directed path would
  go through a set of nodes of size strictly smaller than $|\PA(Y)|$.)
  Furthermore, $C_k$ has rank at least $|\PA(Y)|$. (Indeed, if
  $|\PA_{k-1}| > |\PA(Y)|$, then $C_k$ contains a
  $|\PA(Y)|\times |\PA(Y)|$ submatrix that is equal to the identity;
  if $|\PA_{k-1}| = |\PA(Y)|$, then $C_k$ contains a
  $(|\PA(Y)| -1)\times (|\PA(Y)|-1)$ submatrix that is equal to the
  identity, $\PA_{k} \setminus \PA_{k-1} \neq \emptyset$, and the
  entry corresponding to one of the new parents will be non-zero
  $\mu$-almost surely.)  As $W^c$ can be written as a function of
  $\epsilon^X$, the above replacement scheme yields the desired
  form~\eqref{eq:replsc}.  Since
  $C^\times = (C_1 \cdot C_2 \cdot \ldots \cdot C_f)^\top$ and all
  non-zero entries are independent realizations from $\mu$, this
  proves that $C^{\times}$ is $\mu$-almost surely of rank at least
  $|\PA(Y)|$, that is, (A1) holds $\mu$-almost surely.

  Regarding (A2): Proposition~\ref{prop:random_proj} shows that (A2)
  holds $\mu$-almost surely.

  Regarding (A3): Consider a set $S\subseteq\{1,\ldots,d\}$ with
  $\abs{S}=\abs{\PA(Y)}$ and $S\neq\PA(Y)$.  First, we argue that (B3)
  (i) implies (A3). To see this, assume
  $\AN_I[S]\neq\AN_{I}[\PA(Y)]$.  Without loss of generality assume
  that there is an $i^*$ such that
  $i^* \in \AN_I[S] \setminus \AN_{I}[\PA(Y)]$.  This implies that the
  $i^*$th row of $C_{\PA(Y)}$ is entirely zero.  Moreover, there is a
  node $X^j \in S$ such that $i^* \in \AN_I[\{j\}]$, and therefore the
  entry of the $i^*$th row of $C_S$ that corresponds to $X^j$ must be
  non-zero $\mu$-almost surely ($C_{i,j}$ corresponds to the $i$-th
  component of the total causal effect from $I$ on $X^j$ in the SCM
  given in \eqref{eq:scm}). It therefore follows that $\mu$-almost
  surely it holds that
  \begin{equation}
    \label{eq:neededA2_part2}
    \im{C_{S}}\neq\im{C_{\PA(Y)}}.
  \end{equation}
    
  Now consider a set $S$ and assume that (B3) (i) does not hold but
  (B3) (ii) holds. To argue that (A3) holds, we proceed similarly as
  in the part of the proof showing that $(B1)$ implies $(A1)$.  We
  consider the graph $\mathcal{G}$ over the nodes
  $W_{S \cup \PA(Y)}:=\AN(\PA(Y) \cup S) \cap \DE(I)$.  As before, we
  construct a causal order and substitute the nodes one after each
  other. This time, we obtain the equation
  $$
  X^{S \cup \PA(Y)} = C_1 \cdot C_2 \cdot \ldots \cdot C_{f'} \cdot I + h^\times(H, \varepsilon^X)
  $$
  and
  $C_{\cdot, S\cup \PA(Y)} = (C_1 \cdot C_2 \cdot \ldots \cdot
  C_{f'})^\top$. With the same argument as above, we conclude that
  $\mu$-almost surely, the rank of $C_{\cdot, S\cup \PA(Y)}$ is
  strictly larger than $|\PA(Y)|$.  This implies that
  $\im{C_{S}}\neq\im{C_{\PA(Y)}}$ $\mu$-almost surely.  (Indeed, if
  $\im{C_{S}}=\im{C_{\PA(Y)}}$, then each column of $C_{S}$ can be
  written as a linear combination of the columns of $C_{\PA(Y)}$,
  which implies that $C_{S \cup \PA(Y)}$ is of rank at most
  $|\PA(Y)|$.)  This completes the proof of
  Theorem~\ref{thm:sparse_identifiability_graph}.
\end{proof}

\section{Algorithm and Consistency} \label{sec:method}

The theoretical identifiability results from the previous sections
highlight that the causal coefficient $\beta^*$ can be identifiable
even in cases that are considered non-identifiable in classical IV
literature.  We now propose an estimation procedure called \spaceIV
(\textbf{spa}rse \textbf{c}ausal \textbf{e}ffect \textbf{IV}) that
allows us to infer $\beta^*$ from a finite data set
$(X, I, Y)\in\R^{n\times d} \times\R^{n\times m}\times\R^n$.  The
procedure is based on the optimization problem
$\min_{\beta\in\mathcal{B}}\,\norm{\beta}_0$. It iterates over the
sparsity level $s$ and searches over all subsets
$S\subseteq\{1,\ldots,S\}$ of predictors for that sparsity level to
check whether there is a $\beta\in\R^d$ with $\supp(\beta)=S$ that
solves \eqref{eq:moment_eq}.  We motivate our estimator by considering
a hypothesis test. To obtain finite sample guarantees for the test, we
assume that the error term is normally distributed (to obtain
asymptotic results \citep{anderson1950asymptotic}, such assumptions
can be relaxed).

Let us consider a fixed sparsity level $s\in\{1,\ldots,d\}$ and the
null hypothesis
\begin{equation*}
    H_0(s):\quad
    \exists \beta\in\R^d \text{ with $\|\beta\|_0=s$ such that } \beta\in\mathcal{B}.
\end{equation*}
This hypothesis can be tested using the Anderson-Rubin test
\citep{anderson1949estimation}. Let
$P_I\coloneqq I(I^{\top}I)^{-1}I^{\top}$, then the Anderson-Rubin test
statistic is defined as
\begin{equation}
  \label{eq:AR_test_stat}
  T(\beta)\coloneqq\frac{(Y-X\beta)^{\top}P_I(Y-X\beta)}{(Y-X\beta)^{\top}(\vI-P_I)(Y-X\beta)}\frac{n-m}{m},
\end{equation}
and (still,  for Gaussian error variables) satisfies $T(\beta)\sim F_{n-m,m}$, i.e.\ an $F$ distribution with $n-m$ and $m$
degrees of freedom, 
for all
$\beta\in\mathcal{B}$,
see also \citet{Jakobsen2020}. It is known
\citep[e.g.,][]{dhrymes2012econometrics} that the limited maximum
likelihood estimator (LIML) minimizes this test statistic. For any set
$S\subseteq\{1,\ldots,d\}$, denote by
$\hat{\beta}_{\operatorname{LIML}}(S) \in \mathbb{R}^d$ the LIML
estimator based on the subset of predictors $X^S$ (adding zeros in the
other coordinates). It then holds for all $\beta\in\R^d$ with
$\supp(\beta)=S$ that
\begin{equation}
\label{eq:test_stat_prop}
    T(\hat{\beta}_{\operatorname{LIML}}(S))\leq T(\beta).
\end{equation}
Next, for each sparsity level $s\in\{1,\ldots,d\}$ define
\begin{equation}
    \hat{\beta}(s)\coloneqq \hat{\beta}_{\operatorname{LIML}}\left(\argmin_{S\subseteq \{1,\ldots,d\}: |S|= s} T(\hat{\beta}_{\operatorname{LIML}}(S))\right),
\end{equation}
which can be computed by iterating over all subsets with sparsity
level $s$. Then, by \eqref{eq:test_stat_prop}. the hypothesis test
$\phi_s:\R^{n\times d}\times\R^{n\times
  m}\times\R^n\rightarrow\{0,1\}$ defined by
\begin{equation*}
    \phi_s(X, I, Y)=\mathds{1}(T(\hat{\beta}(s))>F_{n-m,m}^{-1}(1-\alpha))
\end{equation*}
has valid level for the null hypothesis $H_0(s)$ (again, if the error
variables are Gaussian; otherwise it has uniform asymptotic
level given sufficient regularity).

Motivated by this test, we now define our estimator
\texttt{spaceIV}. It iterates over $s\in\{1,\ldots,s_{\max}\}$ and in
each step computes $\hat{\beta}(s)$ by exhaustively searching over all
subsets of size $s$. Then, either $\phi_s$ is accepted and \spaceIV
returns $\hat{\beta}_{\leq s_{\max}} := \hat{\beta}(s)$ as its final
estimator or it continues with $s+1$. If none of the tests are
accepted, the procedure outputs
$\hat{\beta}_{\leq s_{\max}} := \hat{\beta}(s_{\max})$ and a warning
indicating that the model assumptions may be violated. The detailed
procedure is presented in Algorithm~\ref{alg:sparse_iv}.
\begin{algorithm}[ht]
\caption{\spaceIV}
\label{alg:sparse_iv}
\SetAlgoLined
\KwIn{predictors $X\in\R^{n\times d}$, response $Y\in\R^{n}$, instruments $I\in\R^{n\times m}$, sparsity threshold $s_{\max}\in\{1,\ldots, d\}$, significance level $\alpha\in(0,1)$}

Initialize sparsity $s \leftarrow 0$ \\
Initialize test as rejected $\phi \leftarrow 1$ \\
\While{$s < s_{\max}$ and $\phi=1$}{
Update sparsity $s\leftarrow s + 1$\\
Set $\mathbf{S}_s$ to be all subsets in $\{1,\ldots,d\}$ of size $s$\\
\For{$S\in\mathbf{S}_s$}{
Compute LIML-estimator $\hat{\beta}_{\operatorname{LIML}}(S)$\\
Compute test statistic $T(\hat{\beta}_{\operatorname{LIML}}(S))$ in \eqref{eq:AR_test_stat}\\
}
Select $S_{\min}\leftarrow\argmin_{S\in\mathbf{S}_s} T(\hat{\beta}_{\operatorname{LIML}}(S))$\\
Set $\hat{\beta}(s) \leftarrow \hat{\beta}_{\operatorname{LIML}}(S_{\min})$\\
Test whether $H_0(s)$ can be rejected: $\phi\leftarrow \mathds{1}(T(\hat{\beta}(s))>F_{n-m,m}^{-1}(1-\alpha))$
}
Set $\hat{\beta}_{\leq s_{\max}} := \hat{\beta}(s)$

\KwOut{Final estimate $\hat{\beta}_{\leq s_{\max}}$ and test result $\phi$}
\end{algorithm}

The proposed \spaceIV estimator $\hat{\beta}_{\leq s_{\max}}$ satisfies the following guarantees.
\begin{theorem} \label{thm:cons} Consider i.i.d.\ data from a data
  generating process of the form~\eqref{eq:scm} for which
  $g(H, \epsilon^Y)$ is Gaussian, $(I, X, Y)$ has mean zero and (A1) and (A2) hold.  Let
  $s_{\max} \in \mathbb{N}$ be such that
  $s_{\max} \geq \|\beta^*\|_0$. Then, the following two statements
  hold.  (i) We have
  $$
  \lim_{n\rightarrow\infty} P(\|\hat{\beta}_{\leq s_{\max}}\|_0 = \|\beta^*\|_0) \geq 1-\alpha.
  $$
  (ii) If, in addition, (A3) holds, we have, 
  for all $\epsilon>0$ that
  $$
  \lim_{n\rightarrow\infty} P(\|\hat{\beta}_{\leq s_{\max}} - \beta^*\|_2 < \varepsilon) \geq 1-\alpha.
  $$
\end{theorem}
The proof can be found in Appendix~\ref{app:proofcons}.

\subsection{Causal Subset Identifiability} \label{sec:subset_identifiability}

It is possible to identify a subset of the causal parents under even
weaker conditions. This can be done in an idea similar to invariant
causal prediction \citep{Peters2016jrssb}. Define the hypothesis
\begin{equation*}
    H_0(S):\;
    \exists \beta\in\R^d \text{ such that $\supp(\beta)=S$ and } \beta\in\mathcal{B}
\end{equation*}
and the corresponding Anderson-Rubin test
$$
\mathds{1}(T(\hat{\beta}_{\operatorname{LIML}}(S))>F_{n-m,m}^{-1}(1-\alpha)).
$$
We then have the following guarantees.
\begin{proposition} \label{prop:icp-test}
\begin{itemize}
\item[(i)] Consider i.i.d.\ data of $(I, X, Y)$ from a data generating
  process of the form
  $$
  Y := X^{\top}\beta^* + g(H, \epsilon^Y),
  $$
  with $I \independent (H, \epsilon^Y)$, $g(H, \epsilon^Y)$
  Gaussian and $(I, X, Y)$ mean zero.  Then,
  \begin{equation} \label{eq:ICP-test}
    \lim_{n \rightarrow \infty} P\left(\textstyle\bigcap_{\substack{S: |S| = |\PA[Y]| \text{ and }\\
          H_0(S) \text{ accepted} }} S \subseteq
      \PA[Y] \right) \geq 1-\alpha,
  \end{equation}
  where we define the intersection over an empty index set as the empty set.
\item[(ii)] Consider now i.i.d.\ data from a data generating process
  of the form~\eqref{eq:scm} such that $g(H, \epsilon^Y)$ is Gaussian and $(I, X, Y)$ has mean zero.
  If (A1) and (A2) hold, then
  \begin{equation} \label{eq:ICPmin-test} \lim_{n \rightarrow \infty}
    P\left( \textstyle\bigcap_{\substack{S: |S| = M \text{ and }\\
          H_0(S) \text{ accepted} }} S \subseteq \PA[Y]
      \right) \geq 1 - \alpha,
  \end{equation}
  where $M := \min\{|S|: H_0(S) \text{ accepted}\}$.
\end{itemize}
\end{proposition}
The first statement requires the sparsity $\|\beta^*\|_0$ of $\beta^*$
to be known. It still holds when replacing $|\PA[Y]|$ by any $k$ such
that $|\PA[Y]| \leq k \leq d$.  The second statement does not require
knowledge of $\|\beta^*\|_0$ and provides a guarantee when increasing
the subset size until one has found a set that is accepted.  The proof
of Proposition~\ref{prop:icp-test} can be found in
Appendix~\ref{app:prop:icp-test}.

\begin{remark}[Allowing for children of $Y$] \label{rem:children}
  We now discuss the scenario where some of the covariates are causal
  descendants of the response $Y$.  More precisely, we extend the
  model in \eqref{eq:scm} to
  \begin{equation*}
    \begin{split}
      X &:= BX + \gamma Y + AI + h(H, \epsilon^X)\\ 
      Y &:= {\beta^*}^\top X + g(H, \epsilon^Y),
    \end{split}
  \end{equation*}
  where we assume that the matrix
  \begin{equation*}
    B_{\text{ext}}\coloneqq
    \begin{pmatrix}
      B & \gamma\\
      (\beta^*)^{\top} &0
    \end{pmatrix}
    \in\R^{(d+1)\times (d+1)}
  \end{equation*}
  is invertible. Theorem~\ref{thm:sparse_identifiability} and
  therefore also the results in Sections~\ref{sec:method} and
  \ref{sec:subset_identifiability} still hold when using
  $C_{\text{ext}}\coloneqq (A^{\top}, 0)(\vI -
  B_{\text{ext}})^{-\top}_{1:d, \cdot}$ instead of $C$. Assumption
  (A2), however, becomes rather restrictive: If there is a child of
  $Y$ such that all directed paths from $I$ to that child go through
  $Y$, (A2) is not satisfied as the exact intervention effect on $Y$
  is recoverable from that child.  In particular, in this generalized
  setting, Proposition~\ref{prop:random_proj} or (B1) and (B2) no
  longer imply that (A2) holds almost surely.
\end{remark}

\section{Numerical experiments}
\label{sec:experiments}

For the numerical experiments, we consider models of the form
\eqref{eq:scm} with $h(H, \epsilon^X)=H+\epsilon^X$,
$g(H, \epsilon^Y)=H+\epsilon^Y$ and dimensions $d=20$, $q=1$ and
$m=10$. We generate $2000$ random models of this form using the
following procedure:
\begin{itemize}[noitemsep,topsep=0pt,parsep=0pt,partopsep=0pt]
\item Generate a random matrix $B\in\R^{20\times 20}$ by sampling a
  random causal order over $X^1,\ldots,X^{20}$. $B$ then has a
  zero-structure that corresponds to a fully connected graph with this
  causal order. Each non-zero entry in $B$ is drawn independently and
  uniformly from $(-1.5, -0.5)\cup(0.5,1.5)$. Finally, each row of $B$
  is rescaled by the maximal value in each row (using one if it is a
  zero row).
\item Generate a random matrix $A\in\R^{20\times 10}$ by sampling each
  entry independently with distribution $\text{Bernoulli}(1/10)$ and
  setting all diagonal entries to $1$.
\item Generate the parameter $\beta^*\in\R^{20}$ by sampling two
  random coordinates uniformly from $\{1,\ldots,d\}$ and setting them
  to $1$. All remaining coordinates are set to zero.
\item The random variables $I$, $H$, $\epsilon^X$ and $\epsilon^Y$ are
  all drawn as i.i.d.\ standard normal.
\end{itemize}
For each random model we sample $6$ data sets with sample sizes
$n\in\{50, 100, 200, 400, 800, 1600\}$. For each data set, we apply
the following four methods: (i) \spaceIV; this is our proposed method
described in Algorithm~\ref{alg:sparse_iv} with $s_{\max}=3$. (ii)
\texttt{OLS-sparse}; this method goes over all subsets of size at most
$s_{\max}$, fits a linear OLS and then selects the subset with the
smallest AIC.  We also compare our estimator to two oracle methods.
(iii) \texttt{oracle-|PA|}; this method iterates over all subsets with
size $2$ (correct parent size), fits the moment
equation~\eqref{eq:moment_eq} and selects the best subset in terms of
a squared loss based on the moment equation. (iv) \texttt{oracle-PA};
this method considers the correct parent set and fits the moment
equation~\eqref{eq:moment_eq}. Each method results in a sparse
estimate $\hat{\beta}$ of $\beta^*$ based on which we compute the root
mean squared error (RMSE) given by $\|\beta^*-\hat{\beta}\|_2$.

\begin{figure}[t]
    \centering
    \includegraphics[width=\linewidth]{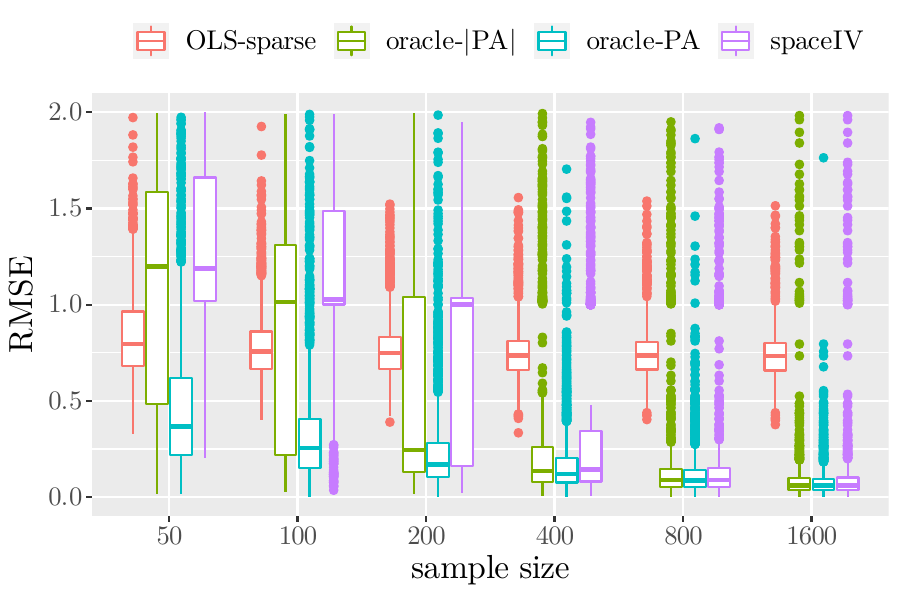}
    \caption{Results for all random models that satisfy (A1)-(A3) (in
      total $1867$ out of $2000$ models). The median RSME of the
      \spaceIV estimator converges to zero as the simple size
      increases, which does not hold for \texttt{OLS-sparse}. Note
      that some of the outliers are cut-off in this plot.}
    \label{fig:consistency_valid}
\end{figure}
For each random model, we explicitly check whether the assumptions
(A1) and (A3) are satisfied by computing $C$ and verifying the
conditions\footnote{Assumption (A2) is satisfied by construction
  because we pick random coefficients for the $B$-matrix, see also
  (B2).}. The results, considering only the random models for which
assumptions (A1)--(A3) are satisfied, are given in
Figure~\ref{fig:consistency_valid}.  As expected, \spaceIV indeed
seems to consistently estimate the causal parameter $\beta^*$, while
\texttt{OLS-sparse} does not. Furthermore, \spaceIV performs worse as
the two oracle methods, illustrating that the estimation in $\spaceIV$
contains three parts: estimating the correct sparsity, estimating the
correct parents set and finally estimating the correct parameters. A
mistake in any of these three steps may result in substantial RMSE,
which explains the outliers in the plot.

To investigate the consistency of estimating the correct sparsity
level in more detail, we consider the fraction of times the correct
sparsity level was selected by \spaceIV. The result is given in
Figure~\ref{fig:expected_sparsity}. It suggests that the sparsity
level is consistently estimated by \spaceIV.
\begin{figure}[ht]
    \centering
    \includegraphics[width=\linewidth]{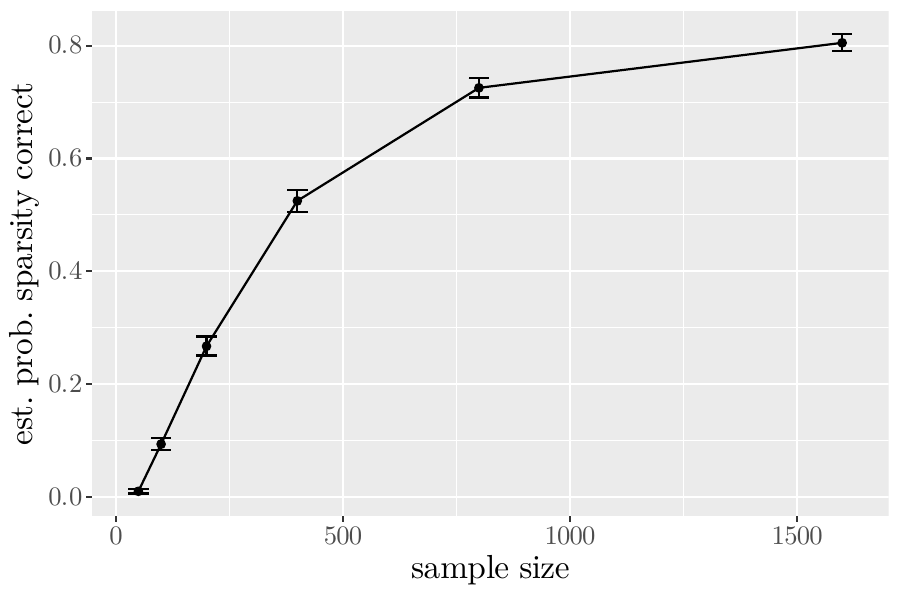}
    \caption{Expected fraction of random models for which \spaceIV
      estimated the correct sparsity level. Only random models that
      satisfy (A1)-(A3) are considered (in total $1867$ models). As
      the sample size increases the estimation of the sparsity level
      becomes more accurate.}
    \label{fig:expected_sparsity}
\end{figure}

Finally, to investigate the performance of \spaceIV based on the
assumptions (A1)--(A3), we compared the performance of all methods at
sample size $n=1600$ depending on which assumptions are satisfied.
(Assumption (A2) is satisfied with probability one, see
Proposition~\ref{prop:random_proj}.) The results are shown in
Figure~\ref{fig:mse_vs_assumptions}. As expected given the theoretical
results presented in Section~\ref{sec:identifsparse}, \spaceIV only
performs well if all assumptions are satisfied. If only assumption
(A1) is satisfied, there are multiple sets with sparsity $2$ for which
the moment equation~\eqref{eq:moment_eq} can be satisfied. Therefore,
while the oracle with the correct parent sets is able to estimate the
causal parameter, \spaceIV and the oracle that only uses the sparsity
level may select wrong sets leading to a larger error. Moreover, if
none of the assumptions are satisfied the causal parameter is not even
identifiable if the true parent set is known.
\begin{figure}[ht]
    \centering
    \includegraphics[width=\linewidth]{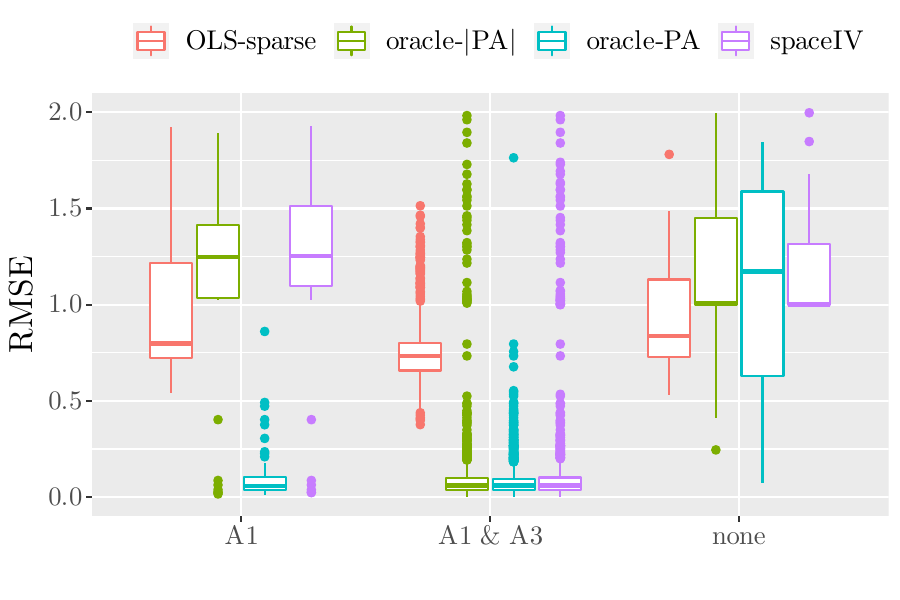}
    \caption{Results for all $2000$ random models with $n=1600$. We
      split the models into three cases depending on which of the
      assumptions (A1) and (A3) are satisfied (the group `(A1)'
      contains $83$ models, the group `(A1) \& (A3)' contains $1867$
      models and the group `none' contains $50$ models). If none of
      the assumptions are satisfied, not even the oracle with known
      parent set works. If only (A1) is satisfied, multiple sets of
      size $2$ are able to satisfy the moment
      equation~\eqref{eq:moment_eq} and \spaceIV may not estimate the
      correct set. These findings are in par with
      Theorem~\ref{thm:sparse_identifiability}.}
    \label{fig:mse_vs_assumptions}
\end{figure}

\section{Conclusion and Future Work}
We have analysed some of the benefits that come with assuming a sparse
causal effect in linear IV models.  We have proved identifiability
results that make the causal effect identifiable even if there are
much less instrument nodes than predictors.  Graphical criteria
provide intuition on these results and characterize for which graphs
the identifiability holds (when randomly choosing coefficients). We
have proposed the estimator $\spaceIV$ and evaluated it on finite
samples.  The results support our theoretical findings and show that
the estimator is often able to find the correct sparsity and the
correct parent set.

We believe that the power result for the Anderson-Rubin test may yield
ways for choosing a significance level for finite samples.
Furthermore, it could be interesting to investigate to which extent
our results generalize to nonlinear models.

\begin{acknowledgements}
  NP was supported by a research grant (0069071) from Novo Nordisk
  Fonden. JP was supported by a research grant (18968) from VILLUM
  FONDEN.
\end{acknowledgements}

\bibliography{uai2022-template.bib}

\newpage

\appendix
\section{Proof of
  Proposition~\ref{thm:partial_identifiability}}\label{app:partial}

\begin{proof}
  Fix $j\in\{1,\ldots,d\}$, then it holds that $\beta^*_j$ is
  identifiable by \eqref{eq:moment_eq} if and only if the space
  $\mathcal{B}$ is degenerate in the $j$-th coordinate, that is,
  $\mathcal{B}_j=\{\beta^*_j\}$.  Next, define $M\coloneqq \cov(I, X)$
  and $v\coloneqq\cov(I, Y)$. Then, denoting the Moore-Penrose inverse
  of $M$ by $M^{\dagger}$, we get that for any solution
  $\beta\in\mathcal{B}$ there exists
  $w\in\nullspace{M} \subseteq \mathbb{R}^d$ such that
  \begin{equation}
    \label{eq:decomposition}
    \beta=M^{\dagger}v + w.
  \end{equation}
  Therefore, the space $\mathcal{B}$ has a degenerate $j$-th
  coordinate if and only if $\nullspace{M}_j=\{0\}$. Using
  $M^{\dagger}$, the null space of $M$ can be expressed as
  \begin{equation*}
    \nullspace{M}=\{(\vI-M^{\dagger}M)w\,\vert\, w\in\R^d\}.
  \end{equation*}
  Next, \eqref{eq:scm} and the assumption of joint independence of
  $I$, $\xi^X\coloneqq h(H, \epsilon^X)$ and
  $\xi^Y\coloneqq g(H, \epsilon^X)$ imply that
  \begin{align*}
    M=\cov[I, X]
    &=\cov\left[I, (\vI-B)^{-1}(AI+\xi^X)\right]\\
    &=\cov[I]A^{\top}(\vI-B)^{-\top}\\
    &=\cov[I]C.
  \end{align*}
  Therefore, using the properties of the Moore-Penrose inverse and that $\cov[I]$ is invertible we get that
  \begin{equation}
    \label{eq:MM_part1}
    M^{\dagger}M=C^{\dagger}\cov[I]^{-1}\cov[I]C.
  \end{equation}
  Hence, we get that $M^{\dagger}M=C^{\dagger}C$ which implies that
  $\nullspace{M}=\nullspace{C}$. This proves the first part of the
  statement. The second part of the proposition uses
  \eqref{eq:decomposition} together with
  $\nullspace{M}=\nullspace{C}$. This completes the proof of
  Proposition~\ref{thm:partial_identifiability}.
\end{proof}

\section{Further Results} \label{app:additional}
\begin{proposition}
  \label{prop:random_proj}
  Let $A\in\R^{n\times m}$ and $B\in\R^{n\times p}$ be two matrices
  satisfying
  \begin{equation*}
    \rank{B}\leq\rank{A}
    \quad\text{and}\quad
    \im{A}\neq\im{B}
  \end{equation*}
  and let $W\in\R^m$ be a random variable with a distribution on
  $\R^m$ that is absolutely continuous with respect to Lebesgue
  measure. Then it holds that
  \begin{equation*}
    P(AW\in\im{B})=0.
  \end{equation*}
\end{proposition}

\begin{proof}
  We begin by showing that
  \begin{equation}
    \label{eq:orth_intersection}
    \im{B}^{\bot}\cap\im{A}\neq\varnothing.
  \end{equation}
  Assume for the sake of contradiction this is not true. Then it would
  hold that $\im{A}\subseteq\im{B}$. Moreover, since
  by assumption $\rank{B}\leq\rank{A}$ this would imply that
  $\im{A}=\im{B}$, which contradicts the assumptions on $A$ and
  $B$. Hence, \eqref{eq:orth_intersection} is true.

  Next, let $b_1,\ldots,b_n\in\R^n$ be an orthogonal basis of $\R^n$ such
  that
  \begin{equation*}
    \operatorname{span}(b_1,\ldots,b_{k})=\im{B}^{\bot}
  \end{equation*}
  and
    \begin{equation*}
    \operatorname{span}(b_{k+1},\ldots,b_{n})=\im{B}.
  \end{equation*}
  Then, for every $\ell\in\{1,\ldots,m\}$ there exits unique
  $\alpha_1^{\ell},\ldots,\alpha_n^{\ell}\in\R$ such that
  \begin{equation*}
    A_{\ell}=\sum_{i=1}^n\alpha_i^{\ell}b_i.
  \end{equation*}
  Furthermore, by \eqref{eq:orth_intersection}, it holds that
  there exists at least one $i^*\in\{1,\ldots,k\}$ and
  $\ell^*\in\{1,\ldots,m\}$ such that $\alpha_{i^*}^{\ell^*}\neq 0$.
  Furthermore, for every $w\in\R^m$ it holds that
  \begin{equation*}
    Aw
    =\sum_{\ell=1}^mw^{\ell}A_{\ell}
    =\sum_{\ell=1}^m\sum_{i=1}^nw^{\ell}\alpha_i^{\ell}b_i
    =\sum_{i=1}^n\left(\sum_{\ell=1}^mw^{\ell}\alpha_i^{\ell}\right)b_i.
  \end{equation*}
  This implies that $Aw\in\im{B}$ if and only if
  $\sum_{\ell=1}^mw^{\ell}\alpha_i^{\ell}=0$ for all
  $i\in\{1,\ldots,k\}$. Using this we get
  \begin{align*}
    P(AW\in\im{B})
    &=P(\forall i\in\{1,\ldots,k\}:\,
      \textstyle\sum_{\ell=1}^mW^{\ell}\alpha_i^{\ell}=0)\\
    &\leq P(\textstyle\sum_{\ell\neq
      \ell^*}W^{\ell}\alpha_{i^*}^{\ell}=W^{\ell^*}\alpha_{i^*}^{\ell^*})\\
    &=0,
  \end{align*}
  where for the last step we used that the distribution of $W$ is
  absolutely continuous with respect to Lebesgue measure.
  This completes the proof of Proposition~\ref{prop:random_proj}.
\end{proof}

\section{Proof of Theorem~\ref{thm:cons}} \label{app:proofcons}

\begin{proof}\footnote{Theorem~\ref{thm:cons} and this proof have been updated to fix a mistake in the accepted UAI version of this manuscript.
    The new version uses similar arguments as in the proof of \citet[Theorem~3.1]{Huang2024}.}
  First, for any collection of i.i.d.\ mean zero random variables $(V_1, W_1),\ldots,(V_n,W_n)$ define
  $\widehat{\cov}(V)\coloneqq\frac{1}{n}\sum_{i=1}^{n}V_i^2$ and $\widehat{\cov}(V,
  W)\coloneqq\frac{1}{n}\sum_{i=1}^{n}V_iW_i$. Next define 
  for all $\beta\in\mathbb{R}^d$ the following quantities
  \begin{align*}
    \pi&\coloneqq \cov(I)^{-1}\cov(I, Y)\\
    \widehat\pi&\coloneqq \widehat\cov(I)^{-1}\widehat\cov(I, Y)\\
    \Pi&\coloneqq \cov(I)^{-1}\cov(I, X)\\
    \widehat\Pi&\coloneqq \widehat\cov(I)^{-1}\widehat\cov(I, X)\\
    \Sigma(\beta)&\coloneqq [\cov(Y-\pi^{\top}I)+\beta^{\top}\cov(X-\Pi^{\top}I)\beta\\
                 &\quad\qquad- 2\cov(Y-\pi^{\top}I, X-\Pi^{\top}I)\beta]\cov(I)^{-1}\\
    \widehat\Sigma(\beta)&\coloneqq [\widehat\cov(Y-\widehat\pi^{\top}I)+\beta^{\top}\widehat\cov(X-\widehat\Pi^{\top}I)\beta\\
                         &\quad\qquad- 2\widehat\cov(Y-\widehat\pi^{\top}I, X-\widehat\Pi^{\top}I)\beta]\widehat\cov(I)^{-1}\\
    t(\beta)&\coloneqq \sqrt{\tfrac{n-m}{m}}\widehat\Sigma(\beta)^{-1/2}(\widehat\pi-\widehat\Pi\beta)\\
    \mu(\beta)&\coloneqq \sqrt{\tfrac{n-m}{m}}\widehat\Sigma(\beta)^{-1/2}(\pi-\Pi\beta).
  \end{align*}
  Then, we can reformulate $T(\beta)$ such that
  \begin{equation*}
    T(\beta)=\tfrac{n-m}{m}(\widehat\pi - \widehat\Pi\beta)^{\top}
    \widehat\Sigma(\beta)^{-1}(\widehat\pi - \widehat\Pi\beta)=\|t(\beta)\|_2^2.
  \end{equation*}
  Moreover, let $\mathcal{C}\subseteq\mathbb{R}^d$ be a compact set and denote by $\beta^*\in\mathcal{C}$ a minimizer of
  $\inf_{\beta\in\mathcal{C}}\|t(\beta)\|_2^2$. Then, we can use standard probability bounds to get for all
  $x\in[0,\infty)$ that
  \begin{align}
    &P\left(\inf_{\beta\in\mathcal{C}}\|t(\beta)\|_2^2\leq x\right)\\
    &=P\left(\|t(\beta^*)-\mu(\beta^*)+\mu(\beta^*)\|_2\leq \sqrt{x}\right)\nonumber\\
    &\leq P\left(\big|\|t(\beta^*)-\mu(\beta^*)\|_2-\|\mu(\beta^*)\|_2\big|\leq \sqrt{x}\right)\nonumber\\
    &\leq P\left(\|t(\beta^*)-\mu(\beta^*)\|_{2}\geq\|\mu(\beta^*)\|_{2}\right)\nonumber\\
    &\qquad+
    P\left(\|\mu(\beta^*)\|_2-\|t(\beta^*)-\mu(\beta^*)\|_2\leq \sqrt{x}\right)\nonumber\\
    &\leq 2 P\left(\|t(\beta^*)-\mu(\beta^*)\|_2\geq \|\mu(\beta^*)\|_2 - \sqrt{x}\right)\nonumber\\
    &\leq 2 P\left(\sup_{\beta\in\mathcal{C}}\|t(\beta)-\mu(\beta)\|_2\geq \inf_{\beta\in\mathcal{C}}\|\mu(\beta)\|_2 - \sqrt{x}\right).\label{eq:upper_bound_infT}
  \end{align}
  Furthermore, it holds that
  \begin{align*}
    &\sup_{\beta\in\mathcal{C}}\|t(\beta)-\mu(\beta)\|_2\\
    &=\|\widehat\Sigma(\beta)^{-1/2}\sqrt{\tfrac{n-m}{m}}((\pi-\Pi\beta)-(\widehat\pi-\widehat\Pi\beta))\|_2\\
    &\leq \|\widehat\Sigma(\beta)^{-1/2}\|_{\operatorname{op}}\sqrt{\tfrac{n-m}{m}}(\|\pi-\widehat{\pi}\|_2 + \|\Pi\beta-\widehat\Pi\beta\|_2)\\
    &\leq \left(\lambda_{\min}(\widehat{\Sigma}(\beta))m\right)^{-1/2}\sqrt{n}(\|\pi-\widehat{\pi}\|_2 + \|\Pi\beta-\widehat\Pi\beta\|_2).
  \end{align*}
  By classical asymptotic theory it can be shown that $\widehat\Sigma(\beta)$ converges to $\Sigma(\beta)$, which by
  assumptions on the model is invertible. Hence, we get that $\sup_{\beta\in\mathcal{C}}\|t(\beta)-\mu(\beta)\|_2$ is
  asympototically bounded in probability.
  Furthermore, using that $T$ does not depend on the scale of $\beta$ it follows together with
  \eqref{eq:upper_bound_infT} that
  \begin{align}
    &P\left(\inf_{\beta:\|\beta\|_0=s}T(\beta)\leq x\right)\nonumber\\
    &=P\left(\inf_{\substack{\beta:\|\beta\|_0=s\\\|\beta\|_2=1}}T(\beta)\leq x\right)\nonumber\\
    &\leq 2 P\left(\sup_{\beta:\|\beta\|_2=1}\|t(\beta)-\mu(\beta)\|_2\geq\right.\nonumber\\
    &\qquad\qquad\qquad\qquad\left.\inf_{\beta:\|\beta\|_0=s}\|\mu(\beta)\|_2 - \sqrt{x}\right).\label{eq:inf_Q_bound}
  \end{align}
  Now we first prove (i). Fix $s \in \mathbb{N}$ such that
  $s < \|\beta^*\|_0$ (if $\|\beta^*\|_0=1$, the proof simplifies and
  one can consider \eqref{eq:part_1} directly). Then, for all $\beta \in \mathbb{R}^d$ such that $\|\beta\|_0 = s$, we
  have by Theorem~\ref{thm:sparse_identifiability} that $\cov\left(I, Y-X^{\top}\beta\right) \neq 0$. Furthermore, using that
  $\pi-\Pi\beta=\cov(I)^{-1}\cov(I, Y-\beta^{\top}X)$ and that
  $\beta \mapsto \|\cov(I)^{-1}\cov\left(I,
    Y-X^{\top}\beta\right)\|_2^2=\|\pi - \Pi\beta\|_2^2$ is a quadratic form, there
  exists $c>0$ such that for all
  $\beta \in \mathbb{R}^d$
  with $\|\beta\|_0 = s$ it holds that $\|\pi-\Pi\beta\|_2>c$.

  As $n$ tends to infinity, it holds that $\inf_{\beta:\|\beta\|_0=s}\|\mu(\beta)\|_2$ diverges to infinity in
  probability, since
  \begin{align*}
      \inf_{\beta:\|\beta\|_0=s}\|\mu(\beta)\|_2
      \geq\inf_{\beta:\|\beta\|_0=s}\sqrt{\tfrac{n-m}{m}}\|\widehat\Sigma(\beta)\|^{-1/2}_{\operatorname{op}}c.
  \end{align*}
  Therefore, by \eqref{eq:inf_Q_bound} it holds that
  \begin{align*}
    &\lim_{n\to\infty} P\left(\varphi_s = 1 \right)\\ 
    &= \lim_{n \to \infty} P\left(\inf_{\beta: \norm{\beta}_0 = s} T(\beta) > F_{n-m,m}^{-1}(1-\alpha)\right) \\ 
    &\geq 1 - \lim_{n\to\infty} 2P\left(\sup_{\beta:\|\beta\|_2=1}\|t(\beta)-\mu(\beta)\|_2\geq\right.\\
    &\qquad\qquad\left.\inf_{\beta:\|\beta\|_0=s}\|\mu(\beta)\|_2 - \sqrt{F_{n-m,m}^{-1}(1-\alpha)}\right)\\
    &= 1.
  \end{align*}
  Since this holds for any $s\in\mathbb{N}$ such that
  $s < \|\beta^*\|_0$, we have
  \begin{align}
    &\lim_{n\rightarrow\infty}P(\|\hat{\beta}_{\leq s_{\max}}\|_0 = \|\beta^*\|_0) \nonumber\\
    &\quad =
      \lim_{n\rightarrow\infty}P\left(\min_{s<\|\beta^*\|_0} \phi_s = 1, 
      \phi_{\|\beta^*\|_0} = 0\right)\nonumber\\
    &\quad =
      \lim_{n\rightarrow\infty}P(\phi_{\|\beta^*\|_0} = 0)\nonumber\\
    &\quad \geq 1-\alpha,\label{eq:part_1}
  \end{align}
  where the last statement follows from 
  the fact that $\phi_s$ has valid level.

  Statement (ii) follows with the same argument noting that for all
  $\varepsilon > 0$ there exists a $c >0$ such that for all
  $\beta \in \mathbb{R}^d$ satisfying $\|\beta\|_0 < \|\beta^*\|_0$ or
  $\|\beta\|_0 = \|\beta^*\|_0$ and
  $\|\beta - \beta^*\|_2 \geq \varepsilon$, we have
  $\|\pi - \Pi \beta\|_2 > c > 0$, again, using
  Theorem~\ref{thm:sparse_identifiability}.  This concludes the proof
  of Theorem~\ref{thm:cons}.
\end{proof}

\section{Proof of
  Proposition~\ref{prop:icp-test}} \label{app:prop:icp-test}

\begin{proof}
  To prove the first statement, we note that
  \begin{align*}
    &\left\{\textstyle\bigcap_{\substack{S: |S| = |\PA[Y]| \text{ and }\\
    H_0(S) \text{ accepted}
    }} S \subseteq \PA[Y]\right\}\\
    &\qquad\qquad\qquad
      \supseteq
      \left\{H_0(\PA[Y]) \text{ accepted}\right\}.
  \end{align*}
  But because 
  $$
  T(\beta^*) \geq 
  T(\hat{\beta}_{\operatorname{LIML}}(\PA[Y])),
  $$
  we have 
  $$
  P\left(H_0(\PA[Y]) \text{ accepted }\right) \geq 1 - \alpha. 
  $$

  To prove the second statement, observe that by the definition of $M$
  it holds that
  $$
  \Big\{ 
  M \geq \|\beta^*\|_0
  \Big\}
  \supseteq
  \left\{\min_{s<\|\beta^*\|_0} \phi_s = 1\right\}
  $$
  and therefore
  \begin{align*}
    &\left\{ \textstyle\bigcap_{\substack{S: |S| = M \text{ and }\\
    H_0(S) \text{ accepted}
    }} S \subseteq \PA[Y] \right\}\\
    &\quad\supseteq
      \Bigg\{\big\{\min_{s<\|\beta^*\|_0} \phi_s = 1\big\}\\
    &\qquad\qquad\quad \cap
      \{T(\hat{\beta}_{\operatorname{LIML}}(\PA[Y]))\leq F_{n-m,m}^{-1}(1-\alpha)\} \Bigg\}.
  \end{align*}
  It follows from the first part of
  Theorem~\ref{thm:sparse_identifiability} that for all
  $\beta \in \mathbb{R}^d$ such that $\|\beta\|_0 < \|\beta^*\|_0$, we
  have $\cov\left(I, Y-X^{\top}\beta\right) \neq 0$.  We can therefore
  apply the same arguments as in Theorem~\ref{thm:cons} to argue that
  for all $s < \|\beta^*\|_0$, we have
  $$
  \lim_{n\rightarrow\infty}P(\phi_s = 1)= 1.
  $$
  The statement then follows from
  $ T(\beta^*) \geq T(\hat{\beta}_{\operatorname{LIML}}(\PA[Y]))$ and
  the fact that the Anderson-Rubin test holds level.  This completes
  the proof of Proposition~\ref{prop:icp-test}.
\end{proof}

\section{Example~\ref{ex:1} continued} \label{app:graphs}

Figure~\ref{fig:ex1-marg} discusses the example graph mentioned in
Example~\ref{ex:1}.
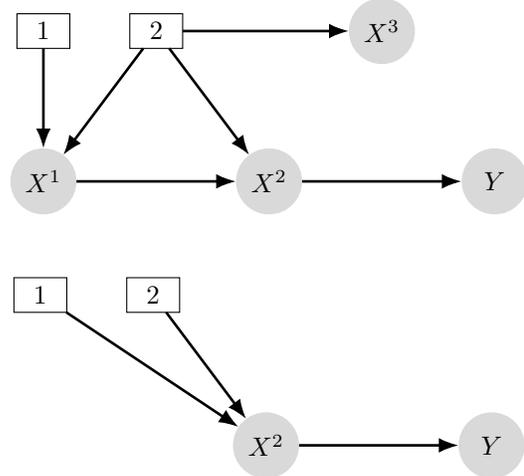
\begin{figure}[ht]
  \centering
  \begin{tikzpicture}[scale=1]
    \tikzstyle{VertexStyle} = [shape = circle, minimum width =
    2.5em, fill=lightgray]
    \Vertex[Math,L=Y,x=0,y=0]{Y}
    \Vertex[Math,L=X^1,x=-6,y=0]{X1}
    \Vertex[Math,L=X^2,x=-3,y=0]{X2}
    \Vertex[Math,L=X^3,x=-1.5,y=2]{X3}
    \tikzstyle{VertexStyle} = [draw, shape = rectangle, minimum
    width=2em]
    \Vertex[Math,L=1,x=-6.0,y=2]{1}
    \Vertex[Math,L=2,x=-4.5,y=2]{2}
    \tikzstyle{VertexStyle} = [draw, dashed, shape = circle, minimum
    width=2.5em]
    \tikzset{EdgeStyle/.append style = {-Latex, line width=1}}
    \Edge(1)(X1)
    \Edge(X1)(X2)
    \Edge(X2)(Y)
    \Edge(2)(X1)
    \Edge(2)(X3)
    \Edge(2)(X2)
  \end{tikzpicture}\vspace{0.8cm}\\
  \begin{tikzpicture}[scale=1]
    \tikzstyle{VertexStyle} = [shape = circle, minimum width =
    2.5em, fill=lightgray]
    \Vertex[Math,L=Y,x=0,y=0]{Y}
    \Vertex[Math,L=X^2,x=-3,y=0]{X2}
    \tikzstyle{VertexStyle} = [draw, shape = rectangle, minimum
    width=2em]
    \Vertex[Math,L=1,x=-6.0,y=2]{1}
    \Vertex[Math,L=2,x=-4.5,y=2]{2}
    \tikzstyle{VertexStyle} = [draw, dashed, shape = circle, minimum
    width=2.5em]
    \tikzset{EdgeStyle/.append style = {-Latex, line width=1}}
    \Edge(1)(X2)
    \Edge(X2)(Y)
    \Edge(2)(X2)
  \end{tikzpicture}
  \caption{Top: Graph copied from Example~\ref{ex:1} and
    Figure~\ref{fig:ex1}.  Assumption (B1) holds because of the path
    $2 \rightarrow X^2$, for example.  For $S=\{1\}$, (B3) (i) is not
    satisfied but (B3) (ii) holds: there is no set $T$ of size one,
    such that all directed paths from $I$ to $\PA(Y)$ go through $T$.
    Therefore, if (B2) holds, the effect $\beta^*$ is identifiable
    (see Theorem~\ref{thm:sparse_identifiability_graph}). If, however,
    we were to remove the second instrument node from
    Example~\ref{ex:1}, (B3)(i) and (ii) would be violated (for set
    $S=\{X^1\}$).  Bottom: Marginalized graph $\mathcal{G}^{\PA(Y)}$.}
  \label{fig:ex1-marg}
\end{figure}

\section{Example violating Assumption (A2)} \label{app:example_A2}

\begin{example} \label{ex:counterA2}
  Consider an SCM of the following form
  {\small
    \begin{align}
      \begin{pmatrix}
        X^1\\
        X^2\\
        X^3
      \end{pmatrix} 
      &:= 
        \begin{pmatrix}
          0 & 0 & 0\\
          0 & 0 & 0\\
          1 & 2 & 0
        \end{pmatrix}
                  \begin{pmatrix}
                    X^1\\
                    X^2\\
                    X^3
                  \end{pmatrix} 
      + 
      \begin{pmatrix}
        4 & 0 \\
        0 & 3 \\
        0 & 0 \\
      \end{pmatrix}
      \begin{pmatrix}
        I^1\\
        I^2
      \end{pmatrix} 
      + h(H, \epsilon^X) 
      \nonumber\\
      Y &:= 
          \begin{pmatrix}
            X^1 & X^2 & X^3
          \end{pmatrix} 
                        \begin{pmatrix}
                          1\\
                          2\\
                          0
                        \end{pmatrix} 
      + g(H, \epsilon^Y), \label{eq:counter_exampleA2}
    \end{align}}
  where   $I^1$, $I^2$, $H$, $\epsilon^Y$, $\epsilon^X$ are jointly independent. 
  Figure~\ref{fig:counter_A2} shows the corresponding graphical representation. In this case, it holds that
  \begin{equation*}
    C=
    \begin{pmatrix}
      1 & 0 & 1\\
      0 &1 & 1
    \end{pmatrix}.
  \end{equation*}
  Hence, the set $S=\{3\}$ violates Assumption (A2). In particular,
  the coefficient $\tilde{\beta}=(0, 0, 1)^{\top}\in\mathcal{B}$
  yields a sparser solution than the causal coefficient
  $(1, 1, 0)^{\top}$. Therefore, the result of
  Theorem~\ref{thm:sparse_identifiability} cannot be valid.
  Assumption (A2) is violated in this example because the coefficients
  can be matched exactly. If the coefficients are chosen randomly with
  a distribution that is absolutely continuous with respect to
  Lebesgue measure, this happens with probability zero, see
  Proposition~\ref{prop:random_proj}.
  \begin{figure}[ht]
    \centering
    \begin{tikzpicture}[scale=1]
      \tikzstyle{VertexStyle} = [shape = circle, minimum width =
      2.5em, fill=lightgray]
      \Vertex[Math,L=Y,x=0,y=0]{Y}
      \Vertex[Math,L=X^1,x=-3,y=2]{X1}
      \Vertex[Math,L=X^2,x=-3,y=-2]{X2}
      \Vertex[Math,L=X^3,x=-6,y=0]{X3}
      \tikzstyle{VertexStyle} = [draw, shape = rectangle, minimum
      width=2em]
      \Vertex[Math,L=1,x=-6.0,y=2]{1}
      \Vertex[Math,L=2,x=-6,y=-2]{2}
      \tikzstyle{VertexStyle} = [draw, dashed, shape = circle, minimum
      width=2.5em]
      \tikzset{EdgeStyle/.append style = {-Latex, line width=1}}
      \Edge[label=4](1)(X1)
      \Edge[label=3](2)(X2)
      \Edge[label=1](X1)(Y)
      \Edge[label=2](X2)(Y)
      \Edge[label=1](X1)(X3)
      \Edge[label=2](X2)(X3)
    \end{tikzpicture}
    \caption{Example graph for which Assumption (A2) can be violated
      if the edge coefficients are fine-tuned to match each other
      exactly.}
  \label{fig:counter_A2}
\end{figure}
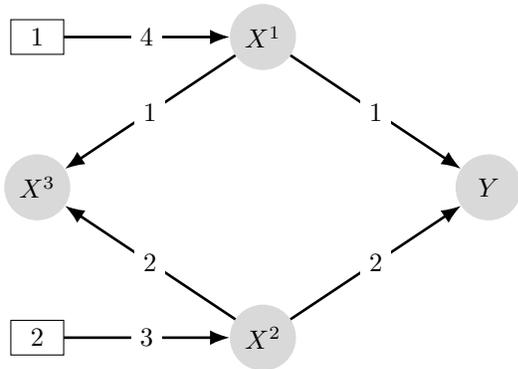
\end{example}

\section{Additional simulation results}

\begin{figure}[t]
    \centering
    \includegraphics[width=\linewidth]{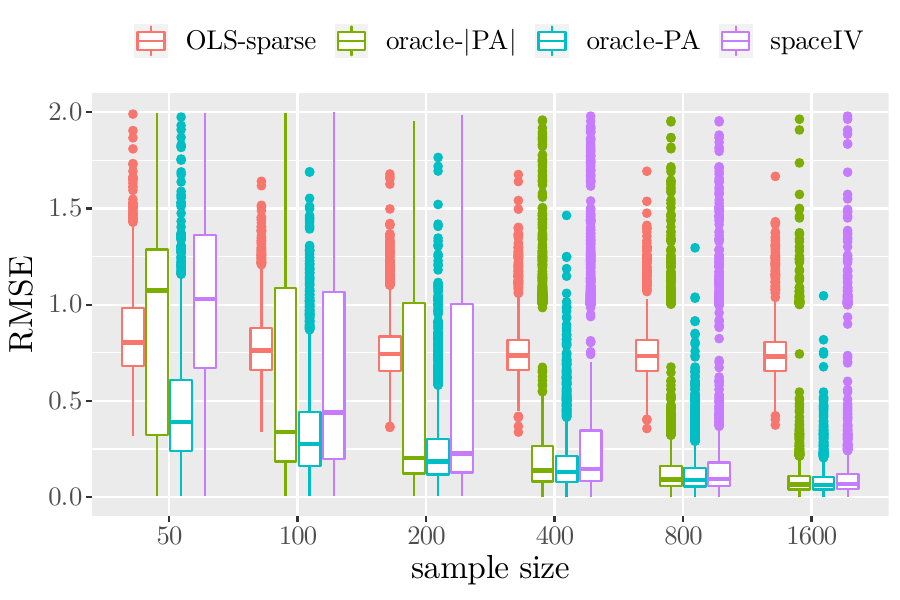}
    \caption{Same experiment as in Figure~\ref{fig:consistency_valid}
      but with TSLS estimator instead of LIML. Results for all random
      models that satisfy (A1)-(A3) (in total $1871$ out of $2000$
      models). The median RSME of the \spaceIV estimator converges to
      zero as the simple size increases, which does not hold for
      \texttt{OLS-sparse}. Note that some of the outliers are cut-off
      in this plot.}
    \label{fig:consistency_valid_TSLS}
\end{figure}

\begin{figure}[ht]
    \centering
    \includegraphics[width=\linewidth]{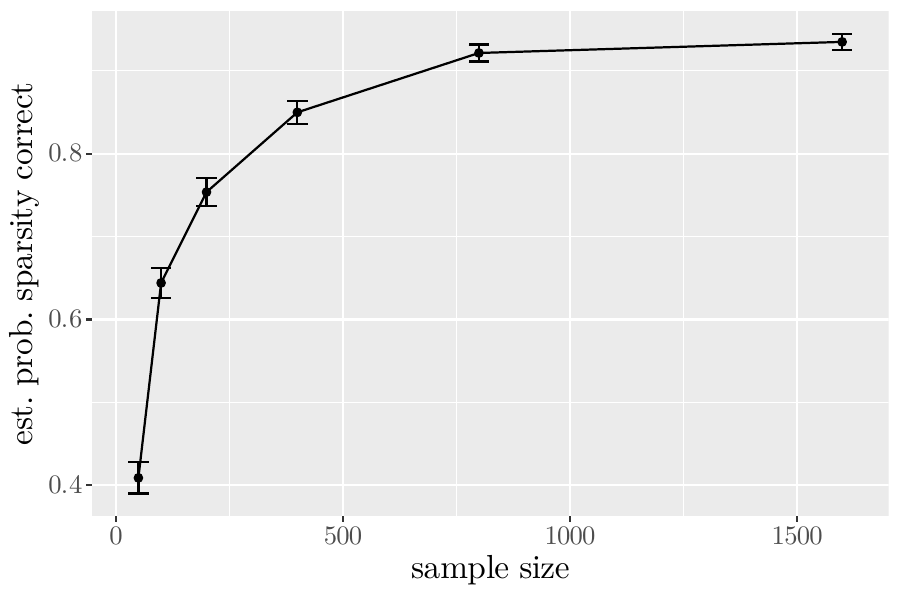}
    \caption{Same experiment as in Figure~\ref{fig:expected_sparsity}
      but with TSLS estimator instead of LIML. Expected fraction of
      random models for which \spaceIV estimated the correct sparsity
      level. Only random models that satisfy (A1)-(A3) are considered
      (in total $1871$ models). As the sample size increases the
      estimation of the sparsity level becomes more accurate.}
    \label{fig:expected_sparsity_TSLS}
\end{figure}

\begin{figure}[ht]
    \centering
    \includegraphics[width=\linewidth]{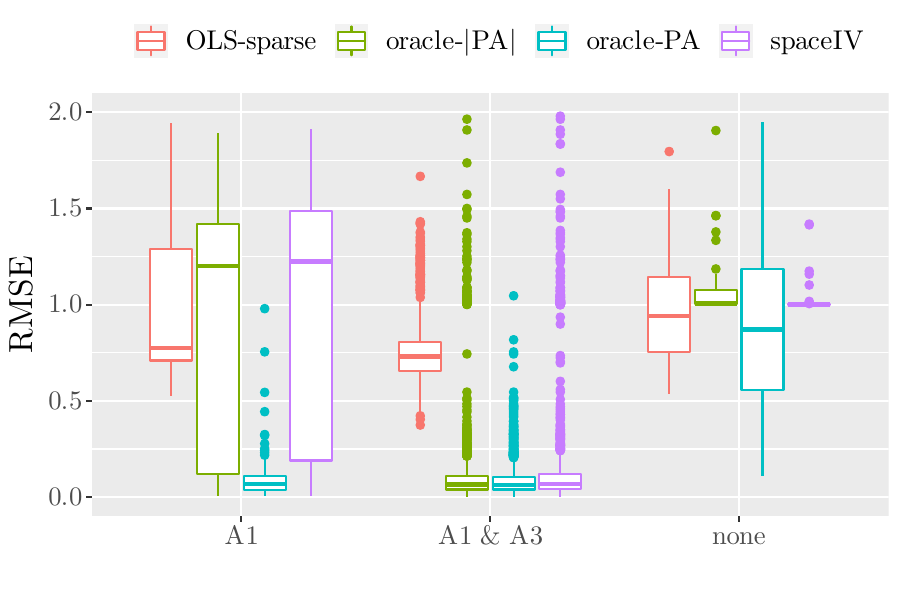}
    \caption{Same experiment as in Figure~\ref{fig:mse_vs_assumptions}
      but with TSLS estimator instead of LIML. Results for all $2000$
      random models with $n=1600$. We split the models into three
      cases depending on which of the assumptions (A1) and (A3) are
      satisfied (the group `(A1)' contains $88$ models, the group
      `(A1) \& (A3)' contains $1871$ models and the group `none'
      contains $41$ models). If none of the assumptions are satisfied,
      not even the oracle with known parent set works. If only (A1) is
      satisfied, multiple sets of size $2$ are able to satisfy the
      moment equation~\eqref{eq:moment_eq} and \spaceIV may not
      estimate the correct set. These findings are in par with
      Theorem~\ref{thm:sparse_identifiability}.}
    \label{fig:mse_vs_assumptions_TSLS}
\end{figure}
\end{document}